\documentclass[11pt,leqno,twoside]{article}

\usepackage[latin1]{inputenc}
\usepackage{amssymb}
\usepackage{amsthm}
\usepackage{bbm,bm}
\usepackage{amsmath}
\usepackage [italian, english]{babel}
\usepackage{algorithm}
\usepackage[noend]{algpseudocode}
\usepackage{graphicx}
\usepackage{booktabs}
\usepackage{caption}
\usepackage{verbatim}

\usepackage[numbers]{natbib}
\setcitestyle{square,aysep={},yysep={,}}

\usepackage[usenames,dvipsnames]{xcolor}
\usepackage{tikz}
\usepackage{tkz-graph}
\usepackage{subfig}
\usepackage{dsfont,fancyhdr}

\allowdisplaybreaks

\graphicspath{{figures/}}
\numberwithin{equation}{section}
\theoremstyle{plain}

\long\def\symbolfootnote[#1]#2{\begingroup\def\thefootnote{\hspace*{-1mm}\fnsymbol{footnote}}\footnote[#1]{#2}\endgroup}

\theoremstyle{definition}
\newtheorem{defin}{Definition}
\theoremstyle{plain}

\theoremstyle{plain}
\newtheorem{prop}{Proposition}
\theoremstyle{plain}

\theoremstyle{plain}
\newtheorem{lemma}{Lemma}
\theoremstyle{plain}

\theoremstyle{plain}

\title{\vspace{-15mm}\bf  
Predictive inference with Fleming--Viot-driven dependent Dirichlet processes
}
\author{
\textsc{Filippo Ascolani}\\
\emph{Bocconi University}\\[2mm]
\textsc{Antonio Lijoi}\\
\emph{Bocconi University and BIDSA}\\[2mm]
\textsc{Matteo Ruggiero}\\
\emph{University of Torino and Collegio Carlo Alberto}
}
\date{\today}

\fancypagestyle{plain}{
\fancyhead{}
}

\usepackage[bookmarks,bookmarksnumbered,colorlinks=true,pdfstartview=FitV, linkcolor=red,citecolor=blue!90!black,urlcolor=blue!90!black]{hyperref}

\bibpunct{\textcolor{blue!90!black}{(}}{\textcolor{blue!90!black}{)}}{\textcolor{blue!90!black}{,}}{a}{\textcolor{blue!90!black}{,}}{\textcolor{blue!90!black}{;}}

\def \d {\mathrm{d}}	
\def\B{\mathcal{B}}

\def\L{\mathcal{L}}
\def\E{\mathbb{E}}
\def\P{\mathbb{P}}
\def\Y{\mathcal{Y}}
\def\Z{\mathcal{Z}}
\def \d {\mathrm{d}}	
\def \mm {\mathbf{m}}	
\def \MM {\mathbf{M}}	
\def \nn {\mathbf{n}}

\def \ee {\mathbf{e}}	
\def \oo {\mathbf{0}}	
\def \YY {\mathbf{Y}}	
\def \dist {K}	

\newcommand\indep{\protect\mathpalette{\protect\independenT}{\perp}} 
\def\independenT#1#2{\mathrel{\rlap{$#1#2$}\mkern4mu{#1#2}}}

\begin{document}

\maketitle
\thispagestyle{empty}

\begin{abstract}
 We consider predictive inference using a class of temporally dependent Dirichlet processes driven by Fleming--Viot diffusions, which have a natural bearing in Bayesian nonparametrics and lend the resulting family of random probability measures to analytical posterior analysis. 
Formulating the implied statistical model as a hidden Markov model, we fully describe the predictive distribution induced by these Fleming--Viot-driven dependent Dirichlet processes, for a sequence of observations collected at a certain time given another set of draws collected at several previous times. This is identified as a mixture of P\'olya urns, whereby the observations can be values from the baseline distribution or copies of previous draws collected at the same time as in the usual P\`olya urn, or can be sampled from a random subset of the data collected at previous times. We characterise the time-dependent weights of the mixture which select such subsets and discuss the asymptotic regimes. We describe the induced partition by means of a Chinese restaurant process metaphor with a \emph{conveyor belt}, whereby new customers who do not sit at an occupied table open a new table by picking a dish either from the baseline distribution or from a time-varying offer available on the conveyor belt. We lay out explicit algorithms for exact and approximate posterior sampling of both observations and partitions, and illustrate our results on predictive problems with synthetic and real data.
\medskip

\medskip

\end{abstract}


\newpage
\section{Introduction and summary of results}

Bayesian nonparametric methodology has undergone a tremendous development in the last decades, often standing out among competitors for flexibility, interpretability and computational convenience. See for example \cite{HHMW,
MQJH15,bGvdV}.
The cornerstone of Bayesian nonparametrics is the sampling model based on the {Dirichlet process} \citep{F73}, whereby
\begin{equation}\label{DPmodel}
\begin{aligned}
& Y_{ i} \mid X = x \overset{\text{iid}}{\sim} x,\quad \quad 
 X \sim \Pi_{\alpha}.
\end{aligned}
\end{equation}
Here, given a sampling space $\Y$, we use $X$ to denote a random probability measure (RPM) on $\Y$, and observations $Y_{i}$ are assumed to be independent with distribution $x$ when $X=x$. 
We also denote by $\Pi_{\alpha}$ the distribution $X$ induces on the space $\mathcal{P}(\Y)$ of probability measures on $\Y$, with $\alpha=\theta P_0$, $\theta > 0$ and $P_{0}$ a nonatomic probability measure on $\Y$. Notable properties of the {Dirichlet process} are its large weak support and conjugacy, whereby the conditional RPM $X$, given observations $Y_{1},\ldots,Y_{n}$ from \eqref{DPmodel}, is still a Dirichlet process with updated parameter $\alpha+\sum_{i=1}^{n}\delta_{Y_{i}}$.

The great appeal offered by the relative simplicity of the {Dirichlet process} boosted a number of extensions, among which some of the most successful are mixtures of Dirichlet processes \citep{A74}, Dirichlet process mixtures \citep{L84}, P\'olya trees \citep{MSW92,L92}, Pitman--Yor processes \citep{PPY92,PY97}, Gibbs-type random measures \citep{GP05,TPAMI}, normalised random measures with independent increments \citep{RLP03,LMP05,LMP07}, to mention a few.
The common thread linking all the above developments is the assumption of exchangeability of the data, equivalent to the conditional independence and identity in distribution in \eqref{DPmodel} 
by virtue of de Finetti's Theorem. This can be restrictive when modelling data   that are known to be generated from partially inhomogeneous sources, as for example in time series modelling or when the data are collected in subpopulations. Such framework can be accommodated by \emph{partial exchangeability}, a weaker type of dependence whereby observations in two or more groups of data are exchangeable within each group but not overall. If groups are identified by a covariate value $z\in \Z$, then observations are exchangeable only if their covariates have the same value. 

One of the most active lines of research in Bayesian nonparametrics in recent years aims at extending the basic paradigm \eqref{DPmodel} to this more general framework. Besides pioneering contributions, recent progresses have stemmed from \cite{ME99}, who called a collection of RPMs $\{X_z, z \in \Z \}$ indexed by a finite-dimensional measurement $z\in \Z$ a \emph{dependent Dirichlet process} (DDP) if each marginal measure $X_{z}$ is a {Dirichlet process} with parameter that depends on $z$.

Here we focus on DDPs with temporal dependence, and replace  $z$ with $t\in [0,\infty)$ representing time. Previous contributions in this framework include {\cite{D06,CDD07,RT08,GS10,CT12,MR16,Cea17,GMR16,CR16,KKK20a}}. Many proposals in this area start from the celebrated stick-breaking representation of the {Dirichlet process} \citep{S94}, whereby $X$ in \eqref{DPmodel} is such that
\begin{equation}\label{SB}
X \overset{d}{=} \sum_{i \geq 0} V_{i}\prod_{j=1}^{i-1}(1-V_{j}) \, \delta_{Y_i}, 
\quad \quad 
V_{i}\overset{\text{iid}}{\sim}\text{Beta}(1,\theta),
\quad \quad 
Y_{i}\overset{\text{iid}}{\sim}P_{0},
\end{equation} 
and the temporal dependence is induced by letting each $V_i$ and/or $Y_i$ depend on time in a way that preserves the marginal distributions.
This approach has many advantages, among which: simplicity and versatility, since inducing dynamics on $V_{i}$ or $Y_{i}$ allows for a variety of solutions;  flexibility, since under mild conditions the resulting processes have large support (cf.~\citealp{BJQ12}); ease of implementation, since strategies for posterior computation based on MCMC sampling are readily available. However, the stick-breaking structure makes the analytical derivation of further posterior information, like for example characterizing the predictive distribution of the observations, often a daunting task. This typically holds for other approaches to temporal Bayesian nonparametric modelling as well. 
Determining explicitly such quantities would not only give a deeper insight into the model posterior properties, which otherwise remain obscure to a large extent, but also provide a further tool for direct application or as a building block in more involved dependent models, whose computational efficiency would benefit from an explicit computation.

In this paper, we provide analytical results related to the posterior predictive distribution of the observations induced by class of temporal DDP models driven by Fleming--Viot processes. The latter are a class of diffusion processes whose marginal values are Dirichlet processes. The continuous time dependence is a distinctive feature of our proposal, compared to the bulk of literature in the area. In particular, here we complement previous work done in \cite{PRS16}, which focussed on identifying the laws of the dependent RPMs involved, by investigating the distributional properties of future observations, characterized as a mixture of P\'olya urn schemes, and those of the induced partitions.

More specifically, in Section \ref{sec: model} we detail the statistical model we adopt, which directly extends \eqref{DPmodel} by assuming a hidden Markov model structure whereby observations are conditionally \emph{iid}  given the marginal value of a Fleming--Viot-driven DDP. We recall some key properties of this model, and include a new result on the weak support of the induced prior.
In Section \ref{sec: results} we present our main results. Conditioning on samples, with possibly different sizes, collected at $p$ times $0 = t_0 < \dots < t_{p-1} = T$, possibly in different amount at different times, we characterise the predictive distribution of a further sequence drawn at time $T+t$. This task can be seen as a dynamic counterpart of obtaining the predictive distribution of $Y_{k+1}|Y_{1},\ldots,Y_{k}$ for any $k\ge1$ in \eqref{DPmodel}, when the RPM $X$ is integrated out, thus yielding
\begin{equation}\label{DPpredictive}
\P(Y_{k+1}\in A| Y_{1},\ldots,Y_{k}) = \frac{\theta}{\theta+k}P_0(A) + \frac{k}{\theta+k} P_k(A),
\end{equation}
for any Borel set $A$ of $\Y$, where $P_{k}$ denotes the empirical distribution of $(Y_1, \dots , Y_k)$. 
In the hidden Markov model framework, we identify the predictive distribution of the DDP at time $T+t$ as a time-dependent mixture of P\'olya urn schemes. This can be thought of as being generated by a latent variable which selects a random subset of the data collected at previous times, whereby every component of the mixture is a classical posterior P\`olya urn conditioned to a different subset of the past data. We characterize the mixture weights, where the temporal dependence arises, and derive an explicit expression for the correlation between observations at different time points. Furthermore, we discuss two asymptotic regimes of the predictive distribution -- as the time index diverges, which recovers \eqref{DPpredictive}, and as the current sample size diverges, which links the sequence at time $T+t$ with its de Finetti measure -- and lay out explicit algorithms for exact and approximate sampling from the predictive. Next, we discuss the induced partition at time $T+t$ and derive an algorithm for sampling from its distribution. The partition sampling process is interpreted as a \emph{Chinese restaurant with conveyor belt}, whereby arriving customers who do not sit at an already occupied table, open a new table by choosing a dish either from the baseline distribution $P_{0}$ or from a temporally dependent selection of dishes that run through the restaurant on a conveyor belt, which in turn depends on past dishes popularity. We defer all proofs to the Supplementary Material. 
Finally, Section \ref{sec: illustration} illustrates the use of our results for predictive inference through synthetic data and through a dataset on the Karnofsky score related to a Hodgkins lymphoma study.


\section{Fleming--Viot dependent Dirichlet processes}\label{sec: model}

We consider a class of dependent Dirichlet processes with continuous temporal covariate. Instead of inducing the temporal dependence through the building blocks of the stick-breaking representation \eqref{SB}, we let the dynamics of the dependent process be driven by a Fleming--Viot (FV) diffusion. 
FV processes have been extensively studied in relation to population genetics (see \cite{EK93} for a review), while their role in Bayesian nonparametrics was first pointed out in \cite{Wea07} (see also \citealp{FRW09}). 
A loose but intuitive way of thinking a FV diffusion is of being composed by infinitely-many probability masses, associated to different locations in the sampling space $\Y$, each behaving like a diffusion in the interval $[0,1]$, under the overall constraint that the masses sum up to 1. In addition, locations whose masses touch 0 are removed, while new locations are inserted at a rate which depends on a parameter $\theta>0$. As a consequence, the random measures $X_{t}$ and $X_{s}$, with $t\ne s$, will share some, though not all, their support points.

The transition function that characterizes a FV process admits the following natural interpretation in Bayesian nonparametrics (cf.~\citealp{Wea07}). Initiate the process at the RPM $X_{0}\sim \Pi_{\alpha}$, and denote by $D_{t}$ a time-indexed latent variable taking values in $\mathbb{Z}_{+}$. Conditional on $D_{t}=m\in \mathbb{Z}_{+}$, the value of the process at time $t$ is a posterior DP $X_{t}$ with law
\begin{equation}\label{FV transition given latent}
X_{t}\mid (D_{t}=m,Y_{1},\ldots,Y_{m})\,\sim \, \Pi_{\alpha+ \sum_{i = 1}^m\delta_{Y_i}}
\quad \quad 
Y_i\mid X_{0}\overset{\text{iid}}{\sim }X_{0}.
\end{equation} 
Here, the realisation of the latent variable $D_{t}$ determines how many atoms $m$ are drawn from the initial state $X_{0}$, to become atoms of the posterior Dirichlet from which the arrival state is drawn. Such $D_{t}$ is a pure-death process, which starts at infinity with probability one and jumps from state $m$ to state $m-1$ after an exponentially distributed waiting time with inhomogenous parameter $\lambda_m = m(\theta+m-1)/2$. The transition probabilities of $D_{t}$ have been computed by \cite{G80,T84}, and in particular
\begin{equation}\label{D_t}
\P(D_{t}=m\mid D_{0}=\infty)=d_{m}(t)
\end{equation} 
where
\begin{equation}\nonumber
d_{m}(t)=\sum_{k=m}^{\infty}e^{-\lambda_{k}t}(-1)^{k-m}\frac{(\theta+2k-1)(\theta+m)_{(k)}}{m!(k-m)!},
\end{equation} 
$\lambda_{k}=k(\theta+k-1)/2$ and 
where $\theta_{(k)}=\theta(\theta-1)\cdots(\theta-k+1)$ is the Pochhammer symbol. Here the fact that $D_{0}=\infty$ almost surely should be understood as an entrance boundary, i.e., the process decreases from infinity at infinite speed so that at each $t>0$ the value of $D_{t}$ is finite.
The unconditional transition of the FV process is thus obtained by integrating $D_{t},Y_{1},\ldots,Y_{D_{t}}$ out of \eqref{FV transition given latent}, leading to
\begin{equation}\label{transition}
P_t(x,\d x') = \sum_{m = 0}^{\infty} d_m(t) \int_{\Y^m}\Pi_{\alpha + \sum_{i = 1}^m\delta_{y_i}}(\d x')x(\d y_1)\cdots x(\d y_m).
\end{equation}
This was first found by \cite{EG93}. 
It is known that $\Pi_{\alpha}$ is the invariant measure of $P_{t}$ if $X_{0}\sim \Pi_{\alpha}$, i.e., if the initial distribution is Dirichlet, in which case all marginal RPMs $X_{t}$ are Dirichlet processes with the same parameter.
In particular, the death process $D_{t}$ determines the correlation between RPMs at different times. Indeed, a larger $t$ implies a lower $m$ with higher probability, hence a decreasing (on average) number of support points will be shared by the random measures $X_{0}$ and $X_{t}$ when $t$ increases. On the contrary, as $t\rightarrow 0$ we have $D_{t}\rightarrow \infty$, which in turn implies infinitely-many atoms shared by $X_{0}$ and $X_{t}$, until the two RPMs eventually coincide.
See \cite{LRS16} for further discussion.

For definiteness, we formalize the following definition.

\begin{defin}\label{fv_def}
A Markov process $\{X_t\}_{t \geq 0}$ taking values in the space of atomic measures on $\Y$ is a \emph{Fleming--Viot dependent Dirichlet process} with parameter $\alpha$, denoted $X_{t}\sim\text{FV-DDP}(\alpha)$,  if $X_0 \sim \Pi_{\alpha}$ and its transition function is \eqref{transition}.
\end{defin}

Seeing a FV-DDP as a collection of RPMs, one is immediately led to wonder about the support properties of the induced prior. The weak support of a DDP indexed by an $\mathbb{R}_{+}$-valued covariate is the smallest closed set in $\B\{ \mathcal{P}(\Y)^{\mathbb{R}_+} \}$ with probability one, where $\mathcal{P}(\Y)$ is the set of probability measures on $\Y$ and $\B\{ \mathcal{P}(\Y)^{\mathbb{R}_+} \}$ is the Borel $\sigma$-field generated by the product topology of weak convergence. \cite{BJQ12} investigated these aspects for a large class of DDPs based on the stick-breaking representation of the {Dirichlet process}. Since no such representation is known for the FV process, our case falls outside that class. The following proposition states that a FV-DDP has full weak support, relative to the support of $P_0$.

\begin{prop}\label{weak_result}
Let $\alpha=\theta P_{0}$ and {$\Y$} be the support of $P_0$. Then the weak support of a \emph{FV-DDP}$(\alpha)$ is given by $\mathcal{P}(\Y)^{\mathbb{R}_+}$.
\end{prop}

In order to formalize the statistical setup, we cast the FV-DDP into a hidden Markov model framework. A hidden Markov model is a double sequence $\{ (X_{t_n},Y_{t_{n}}), n \geq 0 \}$ where $X_{t_n}$ is an unobserved Markov chain, called hidden or \emph{latent signal}, and $Y_{t_n}$ are conditionally independent observations given the signal. The signal can be thought of as the discrete-time sampling of a continuous time process, and is assumed to completely specify the distributions of the observations, called \emph{emission distributions}. While the literature on hidden Markov models has mainly focussed on finite-dimensional signals, infinite-dimensional cases have been previously considered in \cite{BGR02,VSTG08,SGGL09,YPRH11,ZZZ14,PRS16}. 

Here we take $X_{t_{n}}$ to be a FV-DDP as in Definition \ref{fv_def}, evaluated at $p$ times $0 = t_0 < \dots < t_{p-1}=T$. The sampling model {is thus}
\begin{equation}\label{model}
\begin{aligned}
 Y_{t_n}^{i} \mid X_{t_{n}} = x \overset{\text{iid}}{\sim} x, \quad \quad 
 X_t \sim \text{FV-DDP}(\alpha).
\end{aligned}
\end{equation}
It follows that any two variables $Y_{t_{n}}^{i},Y_{t_{m}}^{j}$ are conditionally independent given $X_{t_{n}}$ and $X_{t_{m}}$, with product distribution $X_{t_{n}}\times X_{t_{m}}$.

In addition, similarly to mixing a DP with respect to its parameter measure as in \cite{A74}, one could also consider randomizing {the parameter $\alpha$} in \eqref{model}, e.g.~by letting $\alpha=\alpha_{\gamma}$ and $\gamma\sim \pi$ on an appropriate space.

In the following, we will denote for brevity $\YY_{n}:=\YY_{t_{n}}$ and $\YY_{0:T} := (\YY_{0}, \dots, \YY_T)$, where $\YY_{i}$ is the set of $n_{i}$ observations collected at time $t_{i}$. We will sometimes refer to $\YY_{0:T}$ as the \emph{past values}, since the inferential interest will be set at time $T+t$. We will also denote by $(y^*_1, \dots, y^*_{\dist})$ the {$\dist$ distinct values in $\YY_{0:T}$}, where $K\le \sum_{i=0}^{T}n_{i}$. In this framework, \cite{PRS16} showed that the conditional distribution of the RPM $X_{T}$, given $\YY_{0:T}$, can be written as
\begin{equation}\label{oldresult}
\L(X_{T} | \YY_{0:T})= \sum_{\mm \in \MM}w_{\mm}\Pi _{\alpha + \sum_{j=1}^{\dist}m_j \delta_{y_j^*}}.
\end{equation}
The weights $w_{\mm}$ can be computed recursively as detailed in \cite{PRS16}. In particular, $\MM$ is a finite convex set of vector multiplicities $\mm = (m_1, \dots, m_K) \in \mathbb{Z}_+^{K}$ determined by $\YY_{0:T}$, which identify the mixture components in \eqref{oldresult} with strictly positive weight. 
We will call $\MM$ the set of \emph{currently active indices}. In particular, $\MM$ is given by the points that lie between the counts of $(y^*_1, \dots, y^*_{\dist})$ in $\YY_{T}$, which is the bottom node, and the counts of $(y^*_1, \dots, y^*_{\dist})$ in $\YY_{0:T}$, which is the top node. For example, if $T=1$ suppose we observe $\YY_{0}=(y_{1}^{*},y_{2}^{*})$ for some values $y_{1}^{*}\ne y_{2}^{*}$ and $\YY_{1}=\YY_{0}$, hence $K=2$. 
Then the top node is $(2,2)$ since in $\YY_{0:1}$ there are 2 of each of $(y_{1}^{*},y_{2}^{*})$ and the bottom node is $(1,1)$ which is the counts of $(y_{1}^{*},y_{2}^{*})$ in $\YY_{1}$. Cf.~Figure \ref{fig: graph mixture 0}. Note that observations with $K=3$ distinct values would generate a 3-dimensional graph, with the origin $(0,0,0)$ linked to 3 level-1 nodes $(1,0,0),(0,1,0),(0,0,1)$, and so on. In general, each upper level node is obtained by adding 1 to one of the lower node coordinates.

We note here that the presence of $d_{m}(t)$ in \eqref{transition} makes the computations with FV processes in principle intractable, yielding in general infinite mixtures difficult to simulate from (cf.~\citealp{JS17}). It is then remarkable that conditioning on past data one is able to obtain conditional distributions for the signal given by finite mixtures as in \eqref{oldresult}.

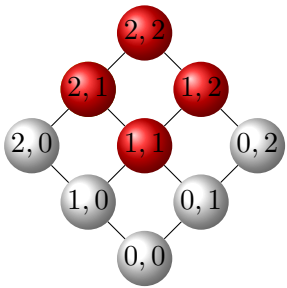
\begin{figure}[t!]
\begin{center} 
 \begin{tikzpicture}[scale=.3]
 \GraphInit[vstyle=Shade]
 \SetGraphShadeColor{yellow}{black}{yellow}
 \tikzset{LabelStyle/.style= {draw,
                              fill  = yellow,
                              text  = black}}
\SetGraphUnit{2.5}
\tikzset{VertexStyle/.style = {shape = circle,
                                shading = ball,
                                ball color = yellow,
                                minimum size = 20pt,
                                inner sep = 1pt,
			       color=black}}
\Vertex[L=\text{$2,1$}]{21}
\tikzset{VertexStyle/.style = {shape = circle,
                                shading = ball,
                                ball color = red,
                                minimum size = 20pt,
                                inner sep = 1pt,
			       color=black}}
\SOEA[L=\text{$1,1$}](21){11}
\NOEA[L=\text{$2,2$}](21){22}
\NOEA[L=\text{$1,2$}](11){12}		       
\SOWE[L=\text{$2,1$}](22){21}		       
\tikzset{VertexStyle/.style = {shape = circle,
                                shading = ball,
                                ball color = white,
                                minimum size = 20pt,
                                inner sep = 1pt,
			       color=black}}
\SOEA[L=\text{$0,1$}](11){01}
\NOEA[L=\text{$0,2$}](01){02}
\SOWE[L=\text{$0,0$}](01){00}
\SOWE[L=\text{$2,0$}](21){20}
\SOEA[L=\text{$1,0$}](20){10}
\tikzset{EdgeStyle/.style = {-,thin}}
\Edge(21)(20)
\Edge(12)(02)
\Edge(20)(10)
\Edge(10)(00)
\Edge(01)(00)
\Edge(02)(01)
\Edge(11)(10)
\Edge(11)(01)
\Edge(22)(12)
\Edge(12)(11)
\Edge(22)(21)
\Edge(21)(11)
\tikzset{VertexStyle/.style = {}}
 \end{tikzpicture}

\begin{quote}
\caption{\scriptsize Red indices in the graph identify active mixture components at time $T$, i.e.~the set $\MM$ in \eqref{oldresult}, corresponding to points $\mm\in \mathbb{Z}_{+}^{K}$ with positive weight. In this example $K=2$, and the graph refers to $\MM$ at time $T=1$ if we observe $\YY_{0}=(y_{1}^{*},y_{2}^{*})=\YY_{1}$.} \label{fig: graph mixture 0}
\end{quote}
\end{center}
\end{figure}



\section{Predictive inference with FV-DDPs}\label{sec: results}

\subsection{Predictive distribution}

In the above framework, we are primarily interested in  predictive inference, which requires obtaining the predictive distribution of $Y_{T+t}^{1},\ldots,Y_{T+t}^{k}|\YY_{0:T}$, that is the marginal distribution of a $k$-sized sample drawn at time $T+t$, given data collected up to time $T$, when the random measures involved are integrated out. See Figure \ref{fig: graphical model}. Note that by virtue of the stationarity of the FV process, if $X_{0}\sim \Pi_{\alpha}$, then $\P(Y_{t}\in A)=P_{0}(A)$ for any $t\ge0$. 
Note also that if one mixes model \eqref{model} by randomizing the parameter measure $\alpha=\alpha_{\gamma}$ as mentioned above, the evaluation the predictive distributions is of paramount importance for posterior computation. Indeed, one needs 
the distribution of {$\gamma |\YY_{0:T}$}, and if for example $\gamma$ has discrete support on $\mathbb{Z}_{+}$ with probabilities $\{p_{j},j\in \mathbb{Z}_{+}\}$, then
\begin{equation}\nonumber
\begin{aligned}
\P(\gamma = j|\YY_{0:T}) \propto p_j \P (\YY_{0:T}| j)
\propto p_j \P ( \YY_{0} | j )\P ( \YY_{1} | \YY_{0}, j ) \cdots \P(\YY_{T} | \YY_{0:T-1},j).
\end{aligned}
\end{equation}
Denote for brevity $Y_{T+t}^{1:k} := (Y_{T+t}^{1}, \dots, Y_{T+t}^{k})$ the $k$ {values} drawn at time $T+t$. For $\mm\in \mathbb{Z}_{+}^{K}$, let $\{\nn\in \mathbb{Z}_{+}^K: \nn \leq \mm\}$ be the set of nonnegative vectors such that $n_i \leq m_i$ for all $i$. Define also $|\nn|:=\sum_{j=1}^{K}n_{i}$, and 
\begin{equation}\label{L(M)}
L(\MM):=\{\nn\in \mathbb{Z}_{+}^{K}: \nn\le \mm, \mm\in \MM\}
\end{equation} 
to be all the points in $\mathbb{Z}_{+}^K$ lying below the top node of $\MM$. E.g., if $\MM$ is given by the red nodes in Figure \ref{fig: graph mixture 0}, then $L(\MM)$ is given by all nodes shown in the figure.

\medskip
\begin{figure}[t!]
\begin{center}
\begin{tikzpicture}[scale=.4]
 \GraphInit[vstyle=Shade]
 \SetGraphShadeColor{yellow}{black}{yellow}
 \tikzset{LabelStyle/.style= {draw,
                              fill  = yellow,
                              text  = black}}
\SetGraphUnit{4}
\tikzset{VertexStyle/.style = {shape = circle,
                                shading = ball,
                                ball color =yellow,
                                minimum size = 30pt,
                                inner sep = 1pt,
                                draw,
			       color=black}}			       
\Vertex[L=\text{$X_{0}$}]{x_0}
\EA[L = \text{$X_{t_1}$}](x_0){x_1}
\EA[L =\dots](x_1){x_dots}
\EA[L = \text{$X_{T}$}](x_dots){x_T}
\EA[L = \text{$X_{T+t}$},unit = 7](x_T){x_t}

\tikzset{VertexStyle/.style = {shape = circle,
                                shading = ball,
                                ball color =green,
                                minimum size = 30pt,
                                inner sep = 1pt,
                                draw,
			       color=black}}			       
\SO[L = \text{$\YY_{0}$}](x_0){y_0}
\EA[L = \text{$\YY_{t_1}$}](y_0){y_1}
\EA[L =\dots](y_1){y_dots}
\EA[L = \text{$\YY_{T}$}](y_dots){y_T}
\tikzset{VertexStyle/.style = {shape = circle,
                                shading = ball,
                                ball color =white,
                                minimum size = 30pt,
                                inner sep = 1pt,
                                draw,
			       color=black}}	
\EA[L = \textbf{?},unit = 7](y_T){y_t}
\tikzset{EdgeStyle/.style = {->, ultra thick}}
\Edge(x_0)(x_1)
\Edge(x_1)(x_dots)
\Edge(x_dots)(x_T)
\Edge(x_T)(x_t)
\Edge(x_0)(y_0)
\Edge(x_1)(y_1)
\Edge(x_dots)(y_dots)
\Edge(x_T)(y_T)
\Edge(x_t)(y_t)
 
 \end{tikzpicture}
 \end{center}
\begin{quote}
\caption{\scriptsize The predictive problem depicted as a graphical model. The upper yellow nodes are nonobserved states of the infinite-dimensional signal, the lower green nodes are conditionally independent observed data whose distribution is determined by the signal, the light gray node is the object of interest.}\label{fig: graphical model}
\end{quote}
\end{figure}
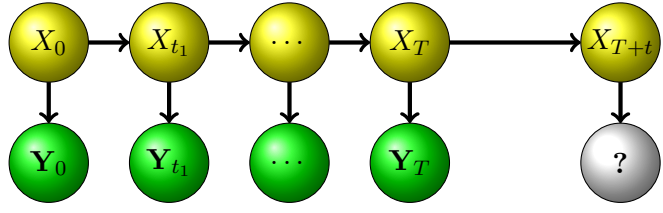

\begin{prop}\label{prop: main}
Assume \eqref{model}, and let the law of $X_{T}$ given data $\YY_{0:T}$ be as in \eqref{oldresult}, where the weights $w_{\mm}$ have been computed recursively. Then, for any Borel set $A$ of $\Y$, the 
first observation at time $T+t$ has distribution
\begin{equation}
\label{predictive 1}
\P 
\,\big(Y_{T+t}\in A|\YY_{0:T}\big) 
=
\sum_{\nn\in L(\MM)}p_t(\MM, \nn) \bigg(
\frac{\theta}{\theta + |\nn|}  P_{0}(A)
+\frac{|\nn|}{\theta + |\nn|}  P_{\nn}(A)
\bigg)
\end{equation}
and the $(k+1)$st observation at time $T+t$, given the first $k$, has distribution
\begin{equation}\label{predictive k+1}
\begin{aligned}
\P &\,\big(Y_{T+t}^{k+1}\in A| \YY_{0:T}, Y_{T+t}^{1:k}\big)=
\sum_{\nn \in L(\MM)}p^{(k)}_t(\MM, \nn)\\
&\times \bigg(
\frac{\theta}{\theta + |\nn| + k}  P_{0}(A)
+\frac{|\nn|}{\theta + |\nn| + k}  P_{\nn}(A)
+\frac{k}{\theta + |\nn| + k}  P_{k}(A)
\bigg)
\end{aligned}
\end{equation} 
where 
\begin{equation}\label{empiricals}
\begin{aligned}
P_\nn =&\,  \frac{1}{|\nn|}\sum_{i=1}^{K}n_i \delta_{y_i^*},\quad \quad 
P_k = \frac{1}{k}\sum_{j = 1}^k \delta_{Y_{T+t}^{j}}
\end{aligned}
\end{equation} 
and $(y^*_1, \dots, y^*_{\dist})$ are the distinct values in $\YY_{0:T}$.
\end{prop}

Before discussing the details of the above statement, a heuristic read of \eqref{predictive 1} is that the first observation at time $T+t$ is either a draw from the baseline distribution $P_{0}$, or a draw from a random subset of the past data points $\YY_{0:T}$, identified by the latent variable $\nn\in L(\MM)$. Given how $L(\MM)$ is defined, $Y_{T+t}$ can therefore be thought of as being drawn from a mixture of P\'olya urns, each conditional on a different subset of the data, ranging from the full dataset to the empty set.  
Indeed, recall from Section \ref{sec: model} that the top node of $\MM$, hence of $L(\MM)$ in \eqref{L(M)}, is the vector of multiplicities of the distinct values $(y_{1}^{*},\ldots,y_{K}^{*})$ contained in the entire dataset $\YY_{0:T}$. The probability weights associated to each lower node $\nn\in L(\MM)$ are determined by a death process on $L(\MM)$, that differs from $D_t$ in \eqref{D_t}. In particular this is a Markov process that jumps from node  $\mm$ to node $\mm-\ee_{i}$ after an Exponential amount of time with parameter $m_{i}(\theta+|\mm|-1)/2$, with $\ee_{i}$ being the canonical vector in the $i$th direction. The weight associated with node $\nn\in L(\MM)$ is then given by the probability that such death process is in $\nn$ after time $t$, if started from any node in $\MM$. For example, if $\MM$ is as in Figure \ref{fig: graph mixture 0}, than the weight of the node $(0,2)$ is given by the probability that the death process is in $(0,2)$ after time $t$ if started from any other node of $\MM$. Being a non increasing process, the admissible starting nodes are $(2,2)$ and $(1,2)$. 
Figure \ref{fig: graph mixture} highlights these two admissible paths of the death process which land at node $(0,2)$.

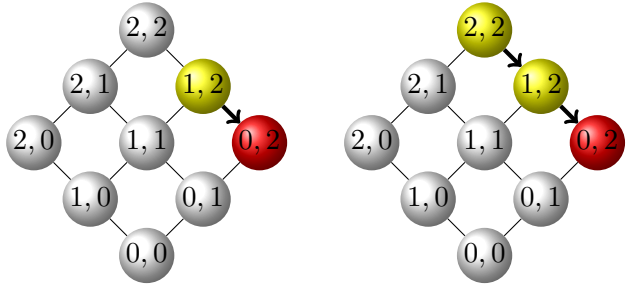
\begin{figure}[t!]
\begin{center} 
   \begin{tikzpicture}[scale=.3]
 \GraphInit[vstyle=Shade]
 \SetGraphShadeColor{yellow}{black}{yellow}
 \tikzset{LabelStyle/.style= {draw,
                              fill  = yellow,
                              text  = black}}
\SetGraphUnit{2.5}
\tikzset{VertexStyle/.style = {shape = circle,
                                shading = ball,
                                ball color = yellow,
                                minimum size = 20pt,
                                inner sep = 1pt,
			       color=black}}
\NOEA[L=\text{$1,2$}](11){12}
\tikzset{VertexStyle/.style = {shape = circle,
                                shading = ball,
                                ball color = red,
                                minimum size = 20pt,
                                inner sep = 1pt,
			       color=black}}
\NOEA[L=\text{$0,2$}](01){02}
\tikzset{VertexStyle/.style = {shape = circle,
                                shading = ball,
                                ball color = white,
                                minimum size = 20pt,
                                inner sep = 1pt,
			       color=black}}
\NOEA[L=\text{$2,2$}](21){22}
\Vertex[L=\text{$2,1$}]{21}
\SOWE[L=\text{$2,0$}](21){20}
\SOEA[L=\text{$0,1$}](11){01}
\SOEA[L=\text{$1,1$}](21){11}
\SOEA[L=\text{$1,0$}](20){10}
\SOWE[L=\text{$0,0$}](01){00}
\tikzset{EdgeStyle/.style = {-,thin}}
\Edge(21)(20)
\Edge(20)(10)
\Edge(10)(00)
\Edge(01)(00)
\Edge(02)(01)
\Edge(11)(10)
\Edge(11)(01)
\Edge(12)(11)
\Edge(22)(21)
\Edge(21)(11)
\Edge(22)(12)
\tikzset{EdgeStyle/.style = {->, ultra thick}}
\Edge(12)(02)
\tikzset{VertexStyle/.style = {}}
 \end{tikzpicture}
 \hspace{5mm}  \begin{tikzpicture}[scale=.3]
 \GraphInit[vstyle=Shade]
 \SetGraphShadeColor{yellow}{black}{yellow}
 \tikzset{LabelStyle/.style= {draw,
                              fill  = yellow,
                              text  = black}}
\SetGraphUnit{2.5}
\tikzset{VertexStyle/.style = {shape = circle,
                                shading = ball,
                                ball color = yellow,
                                minimum size = 20pt,
                                inner sep = 1pt,
			       color=black}}
\NOEA[L=\text{$2,2$}](21){22}
\NOEA[L=\text{$1,2$}](11){12}
\tikzset{VertexStyle/.style = {shape = circle,
                                shading = ball,
                                ball color = red,
                                minimum size = 20pt,
                                inner sep = 1pt,
			       color=black}}
\NOEA[L=\text{$0,2$}](01){02}
\tikzset{VertexStyle/.style = {shape = circle,
                                shading = ball,
                                ball color = white,
                                minimum size = 20pt,
                                inner sep = 1pt,
			       color=black}}
\Vertex[L=\text{$2,1$}]{21}
\SOWE[L=\text{$2,0$}](21){20}
\SOEA[L=\text{$0,1$}](11){01}
\SOEA[L=\text{$1,1$}](21){11}
\SOEA[L=\text{$1,0$}](20){10}
\SOWE[L=\text{$0,0$}](01){00}
\tikzset{EdgeStyle/.style = {-,thin}}
\Edge(21)(20)
\Edge(20)(10)
\Edge(10)(00)
\Edge(01)(00)
\Edge(02)(01)
\Edge(11)(10)
\Edge(11)(01)
\Edge(12)(11)
\Edge(22)(21)
\Edge(21)(11)
\tikzset{EdgeStyle/.style = {->, ultra thick}}
\Edge(22)(12)
\Edge(12)(02)
\tikzset{VertexStyle/.style = {}}
 \end{tikzpicture}
\begin{quote}
\caption{\scriptsize The weight associated to an index $\nn\in L(\MM)$ at time $T+t$ is determined by the probability that the death process reaches $\nn$ from any active index $\mm\in \MM$ at time $T$. For $\MM$ as in Figure \ref{fig: graph mixture 0}, the weight of the mixture component with index $\nn=(0,2)$, i.e., no atoms $y_{1}^{*}$ and 2 atoms $y_{2}^{*}$, is the sum of the probabilities of reaching node $(0,2)$ via the path starting from  $(1,2)$ (left) and from $(2,2)$ (right).}\label{fig: graph mixture}
\end{quote}
\end{center}
\end{figure}

The transition probabilities of this death process are 
 \begin{equation}\label{death process transitions}
p_{\mm, \nn}(t) = p_{|\mm|, |\nn|}(t)\text{HG}(\mm-\nn; \mm, |\mm-\nn|), \qquad \oo\le\nn\le \mm,
\end{equation}
where $\text{HG}(\textbf{i}; \mm, |\textbf{i}|)$ is the multivariate hypergeometric probability function evaluated at \textbf{i}, namely 
\begin{equation}\nonumber
\text{HG}(\textbf{i}; \mm, |\textbf{i}|) = \frac{\binom{\mm_1}{\textbf{i}_1}\dots\binom{\mm_l}{\textbf{i}_l}}{\binom{|\mm|}{|\textbf{i}|}}, \quad l = \text{dim}(\mm)
\end{equation} 
with dim$(\mm)$ denoting the dimension of vector $\mm$, while $p_{|\mm|, |\nn|}(t)$ is the probability of descending from level $|\mm|$ to $|\nn|$ (see Lemma $1$ in the Supplementary Material).
Hence, in general, the probability of reaching node $\nn\in L(\MM)$ from any node in $\MM$ is 
\begin{equation}\label{death process transitions from M}
p_t(\MM,\nn) = \sum_{\mm \in \MM, \mm \geq \nn}w_{\mm}p_{\mm, \nn}(t).
\end{equation} 
In conclusion, with probability $p_t(\MM,\nn)$ the first draw at time $T+t$ will be either from $P_{0}$, with probability $\theta/(\theta+|\nn|)$, or a uniform sample from the subset of data identified by the multiplicity vector $\nn$.

Concerning the general case for the $(k+1)$st observation at time $T+t$, trivial manipulations of \eqref{predictive k+1} provide different interpretative angles. Rearranging the term in brackets one obtains
\begin{equation}\label{structural resemblance}
\begin{aligned}
&\,\frac{\theta_{\nn}}{\theta_{\nn} + k}  P_{0,\nn} 
+\frac{k}{\theta_{\nn} + k}  P_{k}, 
\end{aligned}
\end{equation} 
which bears a clear structural resemblance to \eqref{DPpredictive}. Here
\begin{equation}\nonumber
\theta_{\nn}=\theta+|\nn|, \quad \quad 
P_{0,\nn}:=\frac{\theta}{\theta+|\nn|}P_{0}+\frac{|\nn|}{\theta+|\nn|}P_{\nn}
\end{equation} 
play the role of concentration parameter and baseline probability measure (i.e, the initial urn configuration), respectively. Thus \eqref{predictive k+1} can be seen as a mixture of P\'olya urns where the base measure has a {randomised discrete component} $P_{\nn}$. Unlike in \eqref{DPpredictive}, observations not drawn from empirical measure $P_{k}$ of the current sample can therefore be drawn either from $P_{0}$ or from the empirical measure $P_\nn$, where past observations are assigned multiplicities $\nn$ with probability $p^{(k)}_t(\MM,\nn)$. 

{An alternative interpretation} is obtained by developing the sum in \eqref{predictive k+1} to obtain a single generalised P\'olya urn, written in compact form as
\begin{equation}\label{pred_notation}
\P\big(Y_{T+t}^{k+1}\in \, \cdot \, | \,\YY_{0:T}, Y_{T+t}^{1:k} \big)  = A_kP_0(\cdot) + \sum_{ i= 1}^{\dist}C_{i,k}\delta_{y_i^*}(\cdot) + B_{k}P_k(\cdot)
\end{equation}
where $A$ is a Borel set of $\Y$.
In this case, the first observation is either from $P_{0}$ or a copy of a past value $\YY_{0:T}$, namely
\begin{equation}\nonumber
Y_{T+1}^{1}\sim
\begin{cases}
P_{0} \quad & \text{w.p. } A_{0}  \\
\delta_{y_{i}^{*}} \quad & \text{w.p. } C_{i,0},
\end{cases}
\end{equation} 
while the $(k+1)$st can also be a copy of one of the first $k$ current observations $Y_{T+t}^{1:k}$, namely
\begin{equation}\nonumber
Y_{T+1}^{k+1}\sim
\begin{cases}
P_{0} \quad & \text{w.p. } A_{k}  \\
\delta_{y_{i}^{*}} \quad & \text{w.p. } C_{i,k}\\
P_{k} \quad & \text{w.p. } B_{k}.
\end{cases}
\end{equation} 
The pool of values to be copied is therefore given by past values $\YY_{0:T}$ and current, already sampled observations $Y_{T+t}^{1:k}$.

After each draw, the weights associated to each node need to be updated according to the likelihood that the observation was generated by the associated mixture component, similarly to what is done for mixtures of Dirichlet processes. Specifically, 
\begin{equation}\label{weight update}
p^{(k+1)}_t\left(\MM, \nn \right) \propto p^{(k)}_t(\MM,\nn)
p(y_{T+t}^{k+1}\mid y_{T+t}^{1:k},\nn)
\end{equation}
where 
\begin{equation}\label{k+1 conditional urn predictive}
p(y_{T+t}^{k+1}\mid y_{T+t}^{1:k},\nn):=
\frac{\theta p_0(y^{k+1}_{T+t})+ \sum_{i=1}^{\dist}n_i \delta_{y_i^*}(\{y^{k+1}_{T+t}\}) + \sum_{j = 1}^{k}\delta_{y^j_{T+t}}(\{y^{k+1}_{T+t}\})}{\theta + |\nn| + k} 
\end{equation} 
is the predictive distribution of the $(k+1)$st observation given the previous $k$ and conditional on $\nn$, 
and $p_0$  is the density of $P_0$ with respect to the Lebsegue or the counting measure. 

As a byproduct of Proposition \ref{prop: main}, we can evaluate the correlation between observations at different time points. 

\begin{prop}
For $t,s>0$, let $Y_{t},Y_{t+s}$ be from \eqref{model}. Then
\begin{equation}\nonumber
\mathrm{Corr}(Y_{t},Y_{t+s})= \frac{e^{-\frac{\theta}{2}s}}{\theta+1}.
\end{equation} 
\end{prop}

Unsurprisingly, the correlation decays to 0 as the lag $s$ goes to infinity. 
Moreover, 
\[
\mathrm{Corr}(Y_{t},Y_{t+s}) \to \frac{1}{\theta+1}, \quad \text{as } s \to 0
\]
which is the correlation of two observations from a DP as in \eqref{DPmodel}.

\subsection{Sampling from the predictive distribution}

In order to make Proposition \ref{prop: main} useful in practice, we provide an explicit algorithm to sample from the predictive distribution \eqref{predictive k+1}, which can be useful \emph{per se} or for approximating posterior quantities of interest. 
Exploiting \eqref{structural resemblance} and the fact that \eqref{predictive k+1} can be seen as a mixture of P\'olya urns, we can see $\nn\in \mathbb{Z}_{+}^{K}$ as a latent variable whereby, given $\nn$, sampling proceeds very similarly to a usual P\'olya urn.

Recalling that $|\nn|=\sum_{j=1}^{K}n_{i}$, a simple algorithm for the $(k+1)$st observation would therefore be:
\begin{list}{
$\bullet$
}{\itemsep=1mm\topsep=2mm\itemindent=7mm\labelsep=2mm\labelwidth=0mm\leftmargin=0mm\listparindent=0mm\parsep=0mm\parskip=0mm\partopsep=0mm\rightmargin=0mm\usecounter{enumi}}
\setcounter{enumi}{0}
\item sample $\nn \in L(\MM)$ w.p.~$p^{(k)}_t \left(\MM, \nn \right)$;
\item sample from $P_{0}, P_{\nn}$ or $P_{k}$ with probabilities proportional to $\theta,|\nn|,k$ respectively;
\item update weights $p^{(k)}_t \left(\MM, \nn \right)$ to $p^{(k+1)}_t \left(\MM, \nn \right)$ for each $\nn \in L(\MM)$.
\end{list}
A detailed pseudo-code description is provided in Algorithm \ref{alg: exact}.

\begin{algorithm}[h!]
\caption{Exact sampling from \eqref{predictive k+1}}\label{alg: exact}
\begin{algorithmic}[1]
\State 
\vspace{-4.8 mm}\begin{tabbing}
\textbf{Input:} 
\= - active nodes at time $T$: $\MM$\\
\> - precision parameter: $\theta$ \\
\> - last mixture weights $p^{(k)}_t(\MM,\nn)$, $\nn\in L(M)$\\
\> - past unique observations: $y_1^*, \dots , y_{\dist}^*$\\
\> - current observations: $y_{T+t}^{1}, \dots , y_{T+t}^{k}$
\end{tabbing}
\State \textbf{Sample} $\nn$ w.p. $p^{(k)}_t \left(\MM, \nn \right)$, $\nn \in L(\MM)$\\
 \textbf{Sample} $Y$ from $P_0,P_\nn$ or $P_k$ w.p.~$\frac{\theta}{\theta+|\nn|+k},\frac{|\nn|}{\theta+|\nn|+k},\frac{k}{\theta+|\nn|+k}$ respectively
\State \textbf{Set} $y^{k+1}_{T+t}=Y$ 
\State \textbf{Update parameters}: 
\For {$\nn \in L(\MM)$ and $p(y_{T+t}^{k+1}\mid y_{T+t}^{1:k})$
as in \eqref{k+1 conditional urn predictive}}
\State $p^{(k+1)}_t\left(\MM, \nn \right) = p^{(k)}_t(\MM,\nn)
p(y_{T+t}^{k+1}\mid y_{T+t}^{1:k})$ 
\State Normalize $p^{(k+1)}_t\left(\MM, \nn \right)$
\EndFor
\end{algorithmic}
\label{alg 1}
\end{algorithm}

A possible downside of the above sampling strategy is that when the set $L\left(\MM\right)$ is large, updating all weights may be computationally demanding. Indeed, the size of the set $L(\MM)$ is $|L(\MM)|=\prod_{j=1}^{K}(1+m_{j})$, where $m_{j}$ is the multiplicity of $y_{j}^{*}$ in the data, which can grow considerably with the number of observations (cf.~also Proposition 2.5 in \citealp{PR14}). 
It is however to be noted that, due to the properties of the death process that ultimately governs the time-dependent mixture weights, typically only a small portion of these will be significantly different from zero. Figure \ref{fig: death} illustrates this point by showing the nodes in $\{0,\ldots,50\}$ with weight larger than $0.05$ at different times, if at time 0 there is a unit mass at the node $50$, when $\theta = 1$. A deeper investigation of these aspects in a similar, but parametric, framework, can be found in \cite{KKK20b}.

\begin{figure}[t!]
\centering
   \includegraphics[width=.6\textwidth]{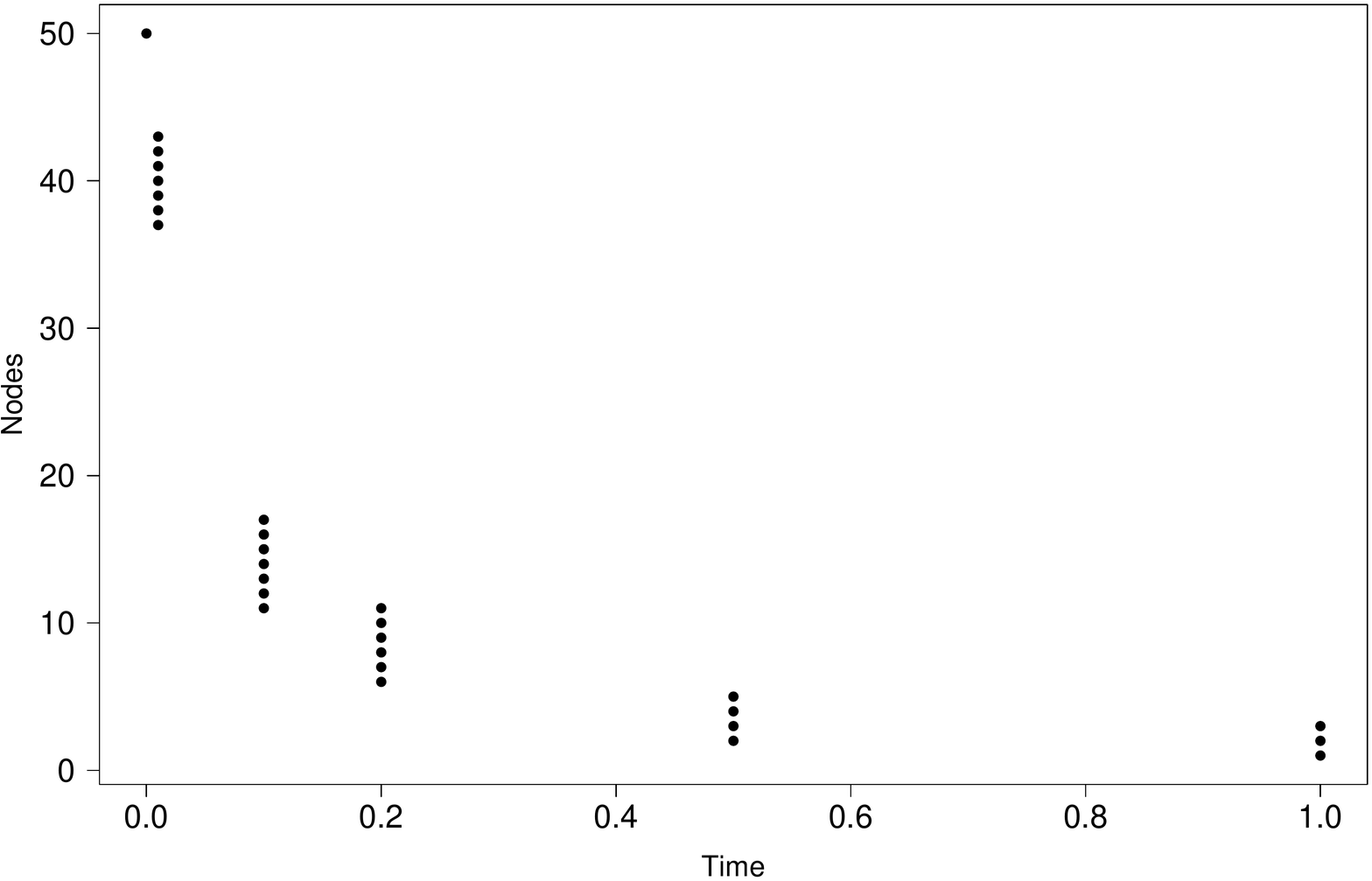}
\begin{quote}
\caption{\scriptsize Nodes in $\{0,\ldots,50\}$ (black dots) with probability of being reached by the death process bigger than $.05$ after lags .01, .1, .2, .5 and 1 (horizontal axis). Starting with mass 1 at the point 50, only a handful of nodes have significant mass after these lags.}
\label{fig: death}
\end{quote}
\end{figure}

Hence an approximate version of the above algorithm can be particularly useful to exploit this aspect. 
We can therefore target a set $\tilde{\MM} \subset L \left( \MM \right)$ such that $|\tilde{\MM}| \ll |L( \MM)|$ and $\sum_{\nn \in \tilde{\MM}}p_t\left(\MM, \nn \right) \approx 1$ by inserting a Monte Carlo step in the algorithm and simulate the death process with a large number of particles. The empirical frequencies of the particles landing nodes will then provide an estimate of the weights $p_t(\MM,\nn)$ in \eqref{predictive 1}. Furthermore, the simulation of the multidimensional death process can be factorised into simulating a one-dimensional death process, which simply tracks the number of steps down the graph, and hypergeometric sampling for choosing the landing node within the reached level. 
A simple algorithm for simulating the death process is as follows:  for $i=1,\ldots,N$,
\begin{list}{
$\bullet$
}{\itemsep=1mm\topsep=2mm\itemindent=-3mm\labelsep=2mm\labelwidth=0mm\leftmargin=9mm\listparindent=0mm\parsep=0mm\parskip=0mm\partopsep=0mm\rightmargin=0mm\usecounter{enumi}}
\setcounter{enumi}{0}
\item draw $\mm$ with probability $w_\mm$ and set $m=|\mm|$;
\item run a one-dimensional death process from $m$, and let $n$ be the landing point after time $t$;
\item draw {$\nn^{(i)}\sim \text{HG}(n,\mm/|\mm|)$;}

\end{list}
and return $\{\nn^{(i)},i=1,\ldots,N\}$.
Note, in turn, that the simulation of the death process trajectories does not require to evaluate its transition probabilities \eqref{death process transitions}, which are prone to numerical instability, and can instead be straightforwardly set up in terms of successive exponential draws by repeating the following cycle: for $i\ge1$,
\begin{list}{
$\bullet$
}{\itemsep=1mm\topsep=2mm\itemindent=7mm\labelsep=2mm\labelwidth=0mm\leftmargin=0mm\listparindent=0mm\parsep=0mm\parskip=0mm\partopsep=0mm\rightmargin=0mm\usecounter{enumi}}
\setcounter{enumi}{0}
\item draw $Z_{i}\sim \text{Exp}(m(\theta+m-1)/2)$
\item if $\sum_{j\le i}Z_{j}<t$ set $m=m-1$ 
else return $n=m-i+1$ and exit cycle.
\end{list}
Algorithm \ref{alg: approx} outlines the pseudocode for sampling approximately from \eqref{predictive k+1} according to this strategy.

\begin{algorithm}[t!]
\caption{Approximate sampling from \eqref{predictive k+1}}\label{alg: approx}
\begin{algorithmic}[1]
\State 
\vspace{-4.8mm}
\begin{tabbing}
\textbf{Input:} 
\= - active nodes at time $T$: $\MM$\\
\> - time to propagate: $t$\\
\> - precision parameter: $\theta$ \\
\> - mixture weights at time $T$: $w_{\mm}$\\
\> - past unique observations: $y_1^*, \dots , y_{\dist}^*$\\
\> - number of Monte Carlo iterates: $N$\\
\end{tabbing}
\vspace{-5.5mm}
\State $\tilde{\MM} = \emptyset$; $w = \emptyset$
\For {$i \in 1:N$ }
\State \textbf{Sample} $\mm$ w.p. $w_{\mm}$, $\mm \in \MM$
\State $n = |\mm|$; $s = t$
\For {$j \geq 1$}
\State \textbf{Sample} $Z$ from Exp$(n(\theta+n-1)/2)$ and set $s = s-Z$
\If {$s > 0$ and $n > 0$}
\State \textbf{Set} $n = n-1$
\Else
\State \textbf{Return} $n$ and exit cycle.
\EndIf
\EndFor
\State \textbf{Sample} $\nn \sim \text{HG}(n,\mm/|\mm|)$
\If {$\nn \not\in \tilde{M}$}
\State \textbf{Add} $\nn$ to $\tilde{\MM}$ and add $1$ to $w$
\Else
\State \textbf{Add} $1$ to the corresponding element of $w$
\EndIf
\EndFor
\State Normalize $w$.
\State Apply algorithm \ref{alg: exact} with $\MM = \tilde{\MM}$ and $p_t(\MM,\nn) = w$
\end{algorithmic}
\end{algorithm}


\subsection{Partition structure and Chinese restaurants with conveyor belt}

A sample from \eqref{predictive k+1} will clearly feature ties among the observations, since there are two discrete sources for the data, namely $P_{\nn}$ and $P_{k}$. A fundamental task concerning sampling models with ties is to characterize the distributional properties of the induced random partition. We say that a random sample $(Y_{1},\ldots,Y_{n})$ induces the partition $(n_{1},\ldots,n_{K})$ if $\sum_{i=1}^{K}n_{i}=n$ and grouping the observed values gives multiplicities $(n_{1},\ldots,n_{K})$. The distribution of a random partition generated by an exchangeable sequence is encoded in the so-called  exchangeable partition probability function, which for the {Dirichlet process} is given by \cite{P06}
\begin{equation}\label{DP EPPF}
p(n_1, \dots , n_k) =  \frac{\theta^k}{\theta_{(n)}}\prod_{i =1}^k(n_i - 1)!.
\end{equation}
The sampling scheme on the space of partitions associated to the Dirichlet process is generally depicted through a Chinese restaurant process \citep{P06}: the first customer sits at a table and orders a dish from the menu $P_{0}$, while successive customers either sit at an existing table $j$, with probability proportional to its current occupancy $n_{j}$, and receive the same dish as the other occupants, or sit at an unoccupied table, with probability proportional to $\theta$, and order from $P_{0}$. 

To account for random partitions induced by a FV-DDP, one can think of a \emph{conveyor belt} typical of some Chinese restaurants, which delivers a non constant selection of dishes that customers can choose to pick up. See  Figure \ref{fig: conveyor belt}. In the context of \eqref{predictive k+1}, each new customer on day $T+t$ faces a different configuration $\nn$ of dishes available on the conveyor belt, determined by the weights $p^{(k)}_t\left(\MM, \nn \right)$. This depends on the following factors: 
(i) which dishes were most popular on day $T$, the greater the popularity, the higher their multiplicity in the nodes of $\MM$, hence the greater their average multiplicity on the conveyor on day $T+t$ as determined by $\nn$;
(ii) the removal of dishes that showed symptoms of food spoilage before the first customer arrives, as determined by the temporal component;
(iii) previous customers choices, as the kitchen readjusts the conveyor at each new customer by reinforcing the most popular dishes, as determined by the update \eqref{weight update}.

\begin{figure}[t!]
\begin{center}
\begin{tikzpicture}[scale=.3]
\draw (-7,2) -- (27,2);
\draw (-7,-2) -- (27,-2);

 \GraphInit[vstyle=Normal]
 \SetGraphShadeColor{yellow}{black}{yellow}
 \tikzset{LabelStyle/.style= {draw,
                              fill  = yellow,
                              text  = black}}
\SetGraphUnit{4}
\tikzset{VertexStyle/.style = {shape = circle,
                                minimum size = 30pt,
                                inner sep = 0pt,
                                draw,
			       color=black}}			       
\Vertex{1}
\EA(1){2}
\EA[L=\text{}](2){3}
\EA[L=\text{1}](3){4}
\EA[L=\text{1}](4){5}
\EA[L=\text{}](5){6}
\draw[->,ultra thick] (-7,0) -- (-3,0);
\draw[->,ultra thick] (23,0) -- (27,0);
\tikzset{VertexStyle/.style = {shape = circle,
                                shading = ball,
                                ball color =white,
                                minimum size = 30pt,
                                inner sep = 1pt,
                                draw,
			       color=black}}	
\tikzset{EdgeStyle/.style = {-, ultra thick}} 
 \end{tikzpicture}
 \end{center}
\begin{quote}
\caption{\scriptsize Schematic depiction of a conveyor belt running through the Chinese restaurant. The conveyor makes available to the customers only a time-varying selection from a pool of dishes. The Figure depicts the current selection, given by three dishes of type 1, one of type 2 and two empty slots from which previously available dishes have been removed.}
\label{fig: conveyor belt}
\end{quote}
\end{figure}
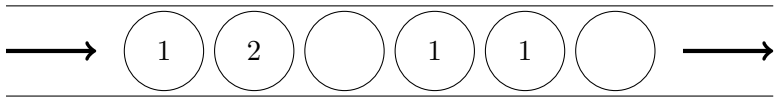

Schematically, the Chinese restaurant process with conveyor belt proceeds as follows. The first customer at time $T+t$ arrives at the restaurant, finds the configuration $\nn$ on the conveyor belt, then picks a dish 
\begin{list}{
$\bullet$
}{\itemsep=1mm\topsep=2mm\itemindent=-3mm\labelsep=2mm\labelwidth=0mm\leftmargin=9mm\listparindent=0mm\parsep=0mm\parskip=0mm\partopsep=0mm\rightmargin=0mm\usecounter{enumi}}
\setcounter{enumi}{0}
\item from the conveyor belt, with probability $|\nn|/(\theta+|\nn|)$ 
\item from the menu $P_{0}$, with probability $\theta/(\theta+|\nn|)$ 
\end{list}
and sits at the first table.
The kitchen then readjusts the offer on the conveyor belt based on the first customer's choice, through \eqref{weight update}. The $(k+1)$st customer arrives at the restaurant, finds a configuration $\nn'$ on the conveyor belt, then
\begin{list}{
$\bullet$
}{\itemsep=1mm\topsep=2mm\itemindent=-3mm\labelsep=2mm\labelwidth=0mm\leftmargin=9mm\listparindent=0mm\parsep=0mm\parskip=0mm\partopsep=0mm\rightmargin=0mm\usecounter{enumi}}
\setcounter{enumi}{0}
\item with probability $m_{j}/(\theta+|\nn'|+k)$ sits at table $j$ and receives the same dish as the other occupants, $m_{j}$ being the current table occupancy
\item otherwise picks a dish
\begin{itemize}
\item from the conveyor belt, with probability $|\nn'|/(\theta+|\nn'|+k)$ 
\item from the menu $P_{0}$, with probability $\theta/(\theta+|\nn'|+k)$ 
\end{itemize} 
and sits at a new table.
\end{list}
Note that node $\oo$ has always positive probability, in which case the conveyor belt is empty and \eqref{predictive k+1} reduces to \eqref{DPpredictive}. Hence a customer facing the configuration $\nn=\oo$ is entering a usual Chinese restaurant.

The usual way for formally deriving the law of a random partition induced by $n$ observations from $X_{T+t}$ would be to compute
\begin{equation}\nonumber
\int_{\Y^{q}}\E\Big[
[X_{T+t}(\d y_{1})]^{n_{1}}\cdots [X_{T+t}(\d y_{q})]^{n_{q}}
\Big], \quad \quad q\le n,
\end{equation} 
which evaluates the probability of all possible configurations of multiplicities $(n_{1},\ldots,n_{q})$, with $q\le n$ and $\sum_{h=1}^{q}n_{h}=n$, irrespective of the values $Y_{i}$ that generated them. This entails a considerable combinatorial complexity, particularly given by the fact that $X_{T+t}$, which has a similar representation to \eqref{oldresult}, is given by a mixture of Dirichlet processes whose base measures have partially shared discrete components. 

Alternatively, one can derive \eqref{DP EPPF} from \eqref{DPpredictive}, better seen by rewriting $P_{k}$ in terms of multiplicities of the distinct values, by assuming observations in the same group arrive sequentially, so that the first group has multiplicity $n_{1}$ with probability proportional to $\theta(n_{1}-1)!$, the second has multiplicity $n_{2}$ with probability proportional to $\theta(n_{2}-1)!$, and so on.
Similarly, we can use the results in Proposition \ref{prop: main} to derive the explicit law of a partition induced by a sample from  $X_{T+t}$. The resulting expression, given in Lemma 2 in the Supplementary Material, suffers from the combinatorial complexity due to the possibility of sampling values that start a group both from $P_{0}$ and from $P_{\nn}$, where $\nn$ is itself random. Here instead we provide an algorithm for generating such random partitions, which can be used, for example, used to study the posterior distribution of the number of clusters directly, i.e.~without resorting to Proposition \ref{prop: main}. 
Mimicking the argument above, we need to  
\begin{list}{
$\bullet$
}{\itemsep=1mm\topsep=2mm\itemindent=-3mm\labelsep=2mm\labelwidth=0mm\leftmargin=8mm\listparindent=0mm\parsep=0mm\parskip=0mm\partopsep=0mm\rightmargin=0mm\usecounter{enumi}}
\setcounter{enumi}{0}
\item choose whether to sample a new value from $P_0$ or from any of the $P_{\nn}$'s
\item draw the new observation after {excluding from $P_{\nn}$ the recorded values}
\item draw the size of the corresponding group.
\end{list}
From \eqref{pred_notation}, the probability of drawing a new observation is therefore given by $A_{k} + \sum_{i \in \mathcal{K}}C_{i,k}$, where $\mathcal{K}$ is the set of past observations still not present in the current sample. The probability of enlarging a group associated to the value $y$ by one is instead
\[
\begin{cases}
B_{k}P_k(\{y\}) \quad \text{if } y  \neq y_j^*, \, \forall j \\
B_{k}P_k(\{y\}) + C_{j,k}  \quad \text{if } y = y_j^*.
\end{cases}
\]
Algorithm \ref{alg: eppf} outlines the pseudocode for sampling a random partition according to this strategy.

\begin{algorithm}[t!]
\caption{Sampling random partitions at time $T+t$}\label{alg: eppf}
\begin{algorithmic}[1]
\State 
\vspace{-4.8 mm}\begin{tabbing}
\textbf{Input:} 
\= - active nodes at time $T$: $\MM$\\
\> - mixture weights at time T: $w_{\mm}$\\
\> - past unique observations: $y_1^*, \dots , y_{\dist}^*$\\
\> - number of observations to draw: $n$
\end{tabbing}
\State \textbf{Initialize} $L= 0$, $\mathcal{K} = \{1, \dots, \dist \}$ and $\mathcal{L} = \emptyset$
\While {$L < n$}
\State \textbf{Sample} $N$ equal to $0$ w.p. $A_L$ and equal to $i$ w.p. $C_{i,L}$, with $i \in \mathcal{K}$.
\If {$ N = 0 $}
\State \textbf{Sample} $Y$ from $P_0$
\State \textbf{Sample} $l$ equal to
\State \quad - $1$ w.p. $A_{L+1} + \sum_{i \in \mathcal{K}}C_{i,L+1}$ 
\State \quad - $j$ w.p. $\left( A_{L+j} + \sum_{i \in \mathcal{K}}C_{i,L+j} \right)\prod_{p=1}^{j-1}B_{L+p}\frac{p}{L+1}$
\State  \quad \quad with $j = \quad 2, \dots , n-L$
\Else
\State \textbf{Set} $Y = y^*_N$ and set $\mathcal{K} = \mathcal{K} \backslash N$
\State \textbf{Sample} $l$ equal to
\State \quad - $1$ w.p. $A_{L+1} + \sum_{i \in \mathcal{K}}C_{i,L+1}$ 
\State \quad - $j$ w.p. $\left( A_{L+j} + \sum_{i \in \mathcal{K}}C_{i,L+j} \right)\prod_{p=1}^{j-1}\left[ C_{i,L+p}B_{L+p}\frac{p}{L+1}\right]$
\State  \quad \quad with $j = \quad 2, \dots , n-L$
\EndIf
\State \textbf{Set} $L = L + l$ and add $Y$ to $\mathcal{L}$.
\EndWhile
\State Return $\mathcal{L}$
\end{algorithmic}
\end{algorithm}


\subsection{Asymptotics}

We investigate two asymptotic regimes for \eqref{predictive k+1}. 
The following Proposition shows that  when $t\rightarrow \infty$, the FV-DDP predictive distribution converges to the usual P\'olya urn \eqref{DPpredictive}.

\begin{prop}\label{lemma_asymp}
Under the hypotheses of Proposition \ref{prop: main}, we have
\[
\L\big(Y_{T+t}^{k+1}|\YY_{0:T}, Y_{T+t}^{1:k}\big) \overset{\mathrm{TV}}{\longrightarrow} \frac{\theta}{\theta+k} P_0  + \frac{k}{\theta+k}P_k, \quad \quad 
a.s., \text{as $t \to \infty$},
\]
with $P_{k}$ as in \eqref{empiricals}.
\end{prop}

Here $\overset{\text{TV}}{\longrightarrow}$ denotes convergence in total variation distance, and the statement is almost sure with respect to the probability measure induced by the FV model on the space of measure-valued temporal trajectories. A heuristic interpretation of the above result is that, when the lag between the last and the current data collection point diverges, the information given by past observations $\YY_{0:T}$ becomes obsolete, and sampling from \eqref{predictive k+1} approximates sampling from the prior P\`olya urn \eqref{DPpredictive}. This should be intuitive, as very old information, relative to the current inferential goals, should have a negligible effect. 

Unsurprisingly, it can be easily proved that an analogous result holds for the distribution of the induced partition, which converges to the EPPF of the {Dirichlet process} as $t\rightarrow \infty$. The proof follows similar lines to that of Proposition \ref{lemma_asymp}. In the conveyor belt metaphor, as $t$ increases all dishes on the conveyor belt have been removed due to food spoilage, before the next customer comes in.

The following Proposition shows that when $k\rightarrow \infty$ in \eqref{predictive k+1}, we recover the law of $X_{T+t}$ given $\YY_{0:T}$ as de Finetti measure.

\begin{prop}\label{lemma_obs}
Under the hypotheses of Proposition \ref{prop: main}, we have
\[
\L\big(Y_{T+t}^{k+1}|\YY_{0:T}, Y_{T+t}^{1:k}\big) \Rightarrow P^{*}_{T+t}, \quad \quad 
a.s., \text{as $k \to \infty$},
\]
where $P^* \sim \L( X_{T+t}|\YY_{0:T} )$.
\end{prop}

Here $P^{*}$ is a random measure with the same distribution as the FV-DDP at time $T+t$ given only the past information $\YY_{0:T}$. Recall for comparison that the same type of limit for \eqref{DPpredictive} yields
\begin{equation}\nonumber
\L(Y_{k+1}| Y_{1},\ldots,Y_{k})\Rightarrow P^{*}, \quad \quad 
P^{*}\sim \Pi_{\alpha},\quad \text{as }k\rightarrow \infty,
\end{equation} 
where $\Pi_{\alpha}$ is the de Finetti measure of the sequence and $P^{*}$ is sometimes called the directing random measure.


\section{Illustration}\label{sec: illustration}

We illustrate predictive inference using FV-DDPs, based on Proposition \ref{prop: main}. Besides the usual prior specification regarding models based on the {Dirichlet process}, that concern the choice of the total mass $\theta$ and of the baseline distribution $P_{0}$, here we can also introduce a parameter $\sigma>0$ that controls the speed of the DDP. This acts as a time rescaling, whereby the data collection times $t_{i}$ are rescaled to $\sigma t_{i}$. This additional parameter provides extra flexibility for estimation, as it can be used to adapt the prior to the correct time scale of the underlying data generating process.

\subsection{Synthetic data}\label{sec: sim data}

We consider data generated by the model
\begin{equation}\nonumber
\begin{aligned}
Y_t \sim&\, \frac{1}{2}\text{Po}(\mu_{t}^{-1},0) + \frac{1}{2}\text{Po}(\nu_{t}^{-1},5),\\
\mu_t =&\, \mu_{t-1}+\varepsilon_{t}, \quad \varepsilon_{t}\sim \text{Exp}(1), \\
\nu_t =&\, \nu_{t-1}+\eta_{t}, \quad \eta_{t}\sim \text{Exp}(1), \eta_t \indep \epsilon_t
\end{aligned}
\end{equation} 
where Po$(\lambda, b)$ denotes a $b$-translated Poisson distribution with parameter $\lambda$, and where $\mu_{0}^{-1} = \nu_{0}^{-1} =5$, for $t = 0,1,2,\dots$.
We collect $15$ observations at each $t\in\{0,\ldots,15\}$ and consider one-step-ahead predictions based on the first 5 and 15 data collection times.

We fit the data by using a FV-DDP model as specified in \eqref{model}, with the following prior specification.
We consider two choices for $P_{0}$, a Negative Binomial with parameters $(2,0.5)$ and a Binomial with parameters $(99,0.3)$, which respectively concentrate most of their mass around small values and around the value 30.  {We consider a uniform prior on $\theta$ concentrated on the points $\{.5,1,1.5, \dots, 15\}$}. A  continuous prior could also be envisaged, at the cost of adding a Metropolis--Hastings step in the posterior simulation, which we avoid here for the sake of computational efficiency.
Similarly, for $\sigma$ we consider a uniform prior on the values $\{ 0.01, 0.1, 0.3, 0.5, 0.7, 0.9, 1.5 \}$. 
The estimates are obtained by means of $500$ replicates of \eqref{predictive k+1} of $1000$ observations each, using the approximate method outlined in Algorithm \ref{alg: approx} with $10000$ Monte Carlo iterates.
We also compare the FV-DDP estimate with that obtained using the DDP proposed in \cite{GMR16}. This is constructed from the stick-breaking representation \eqref{SB} by letting 
\begin{equation}\nonumber
V_{i}(t_{n})\sim c \delta_{V'}+(1-c)\delta_{V_{i}(t_{n-1})}, \quad \quad V^{'}\sim \text{Beta}(1,\theta).
\end{equation} 
in \eqref{SB} and keeping the locations $Y_{i}$ fixed.  We let the resulting DDP be the mixing measure in a time-dependent mixture of Poisson kernels, which provides additional flexibility to this model with respect to our proposal. Furthermore, we give the competitor model a considerable advantage by training it also with the data points collected at times $6$ and $7$, which provide information on the prediction targets, and by centering it on the Negative Binomial with parameters $(2,0.5)$, rather than on the above mentioned mixture, which puts mass closer to where most mass of the true pmf lies.

{Figure \ref{fig:subfig_sigma_prior} shows the results on one-step-ahead prediction with 
15 collection times. 
The posterior of $\sigma$ (not shown) concentrates most of the mass on points 0.7 and 0.9, which leads to learning the correct time scale for prediction, resulting in an accurate estimate of the true pmf. The credible intervals are quite wide, and a better precision may be achieved by increasing the number of time points at which the data are recorded. }

\begin{figure}[h]
\centering
\includegraphics[width=.6\textwidth]{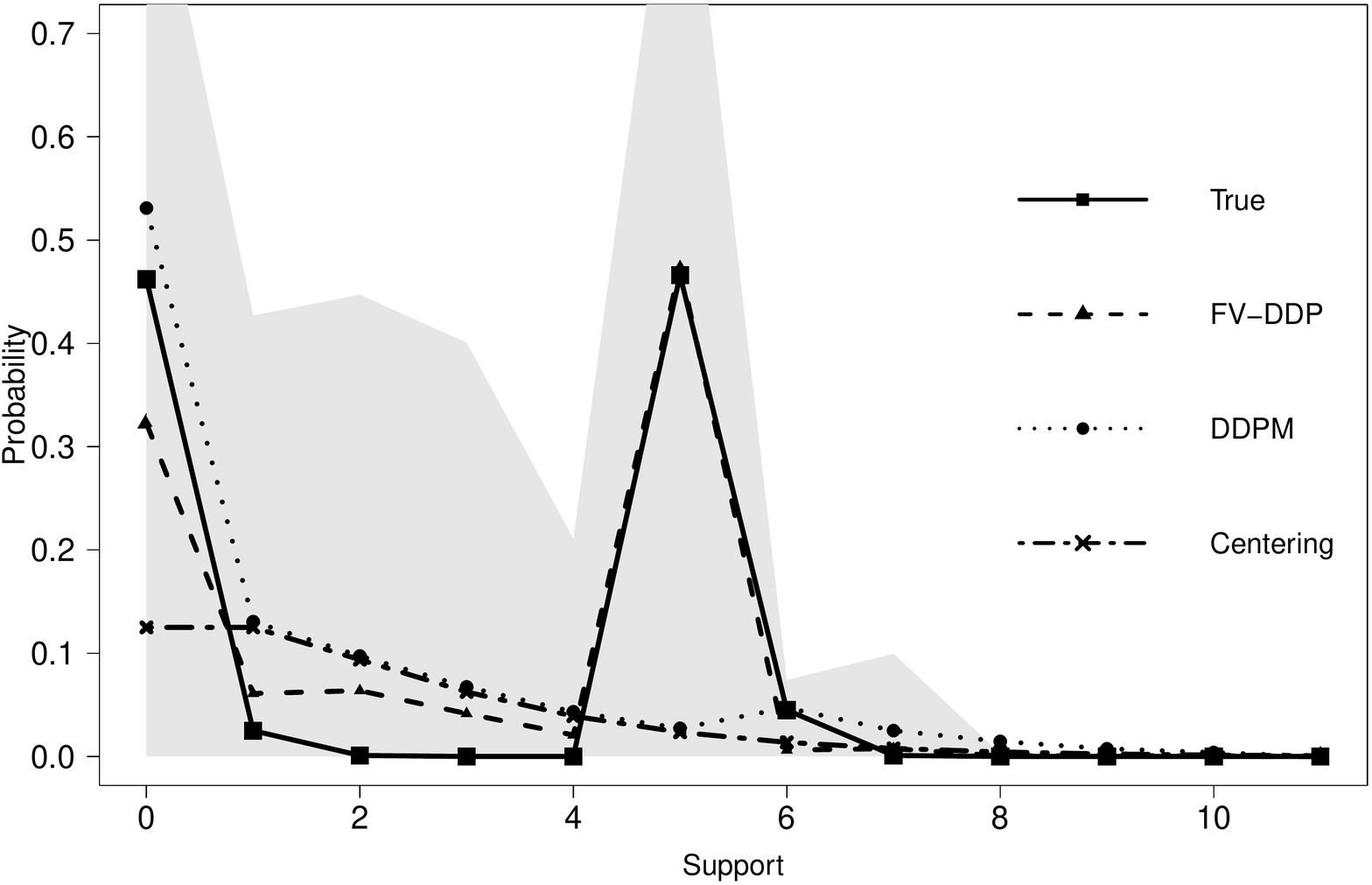}
\begin{quote}
\caption{\scriptsize {One-step-ahead prediction and $95 \%$ pointwise credible intervals, based on 15 data collection times.}}
   \label{fig:subfig_sigma_prior}
\end{quote}
\end{figure}

We compare the previous results with those obtained by choosing $\sigma$ via out-of-sample validation. This is done here using times 0 to 4 as training and time 5 as test, whereby for each $\sigma \in \{.0001,.001,.01,.1,0.5,1,1.5\}$ we compute the sum of absolute errors (SAE) between the FV-DDP posterior predictive mean and the true pmf. These are shown in Table \ref{table 1}, leading to choose $\sigma=.01$. 

\begin{table}[h!]
\begin{center}
\begin{tabular}{cccccccc}\hline
$\sigma$  &.0001 &.001 & .01 & .1 & .5 & 1 & 1.5 \\\hline
\text{SAE}    & .1410 & .1345 & \textbf{.1064} & .1301 & .1261 & .1595 & .1847\\\hline
\end{tabular}
\end{center}
\begin{quote}
\caption{\scriptsize Sum of the absolute error between predicted and true pmf at time $5$ for different values of $\sigma$.}
\label{table 1}
\end{quote}
\vspace{-5mm}
\end{table}

Table \ref{table 2} shows the posterior weights of relevant values of $\theta$ among those with positive prior mass, for the above mentioned choices of $P_{0}$ and using the chosen value of $\sigma$.
The model correctly assigns all posterior probability to the Negative Binomial centering (Binomial not reported in the table), which moves mass towards smaller values as time increases.

\begin{table}[h!]
\begin{center}
\begin{tabular}{ccccc}\hline
$\theta$    & 1 & 1.5 & 2 & 3   \\\hline
\text{NegBinom}    & .5644 & .001694 & .04702 & 0.3868  \\\hline
\end{tabular}
\end{center}
\begin{quote}
\caption{\scriptsize Relevant posterior weights of $\theta$}\label{table 2}
\end{quote}
\vspace{-5mm}
\end{table}

Figure \ref{fig:subfig_1} shows the results in this case for the one- and two-step-ahead predictions given only 5 data collection times. The true pmf is correctly predicted by the FV-DDP estimate even in this short horizon scenario, and the associated $95\%$ pointwise credible intervals are significantly sharper if compared to Figure \ref{fig:subfig_sigma_prior}, obtained with a longer horizon. The prediction based on the alternative DDP mixture does not infer correctly the target, leading to an associated normalised $ \ell_1$ distance from the true pmf of $12.72 \%$ and $12.84\%$, compared to $4.95\%$ and $4.90 \%$ for the FV- DDP prediction.

\begin{figure}[h]
\centering
\includegraphics[width=.49\textwidth]{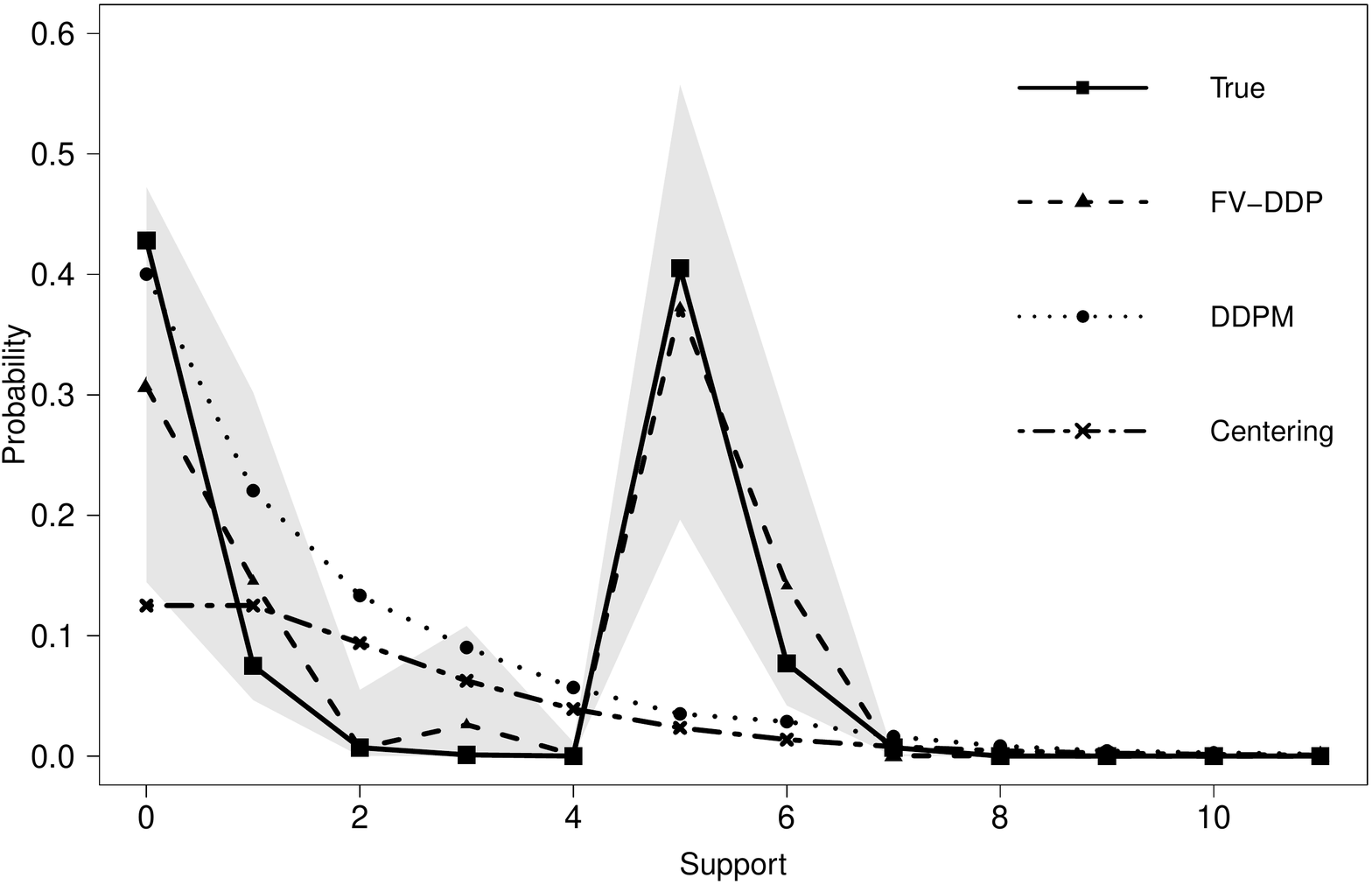}
\includegraphics[width=.49\textwidth]{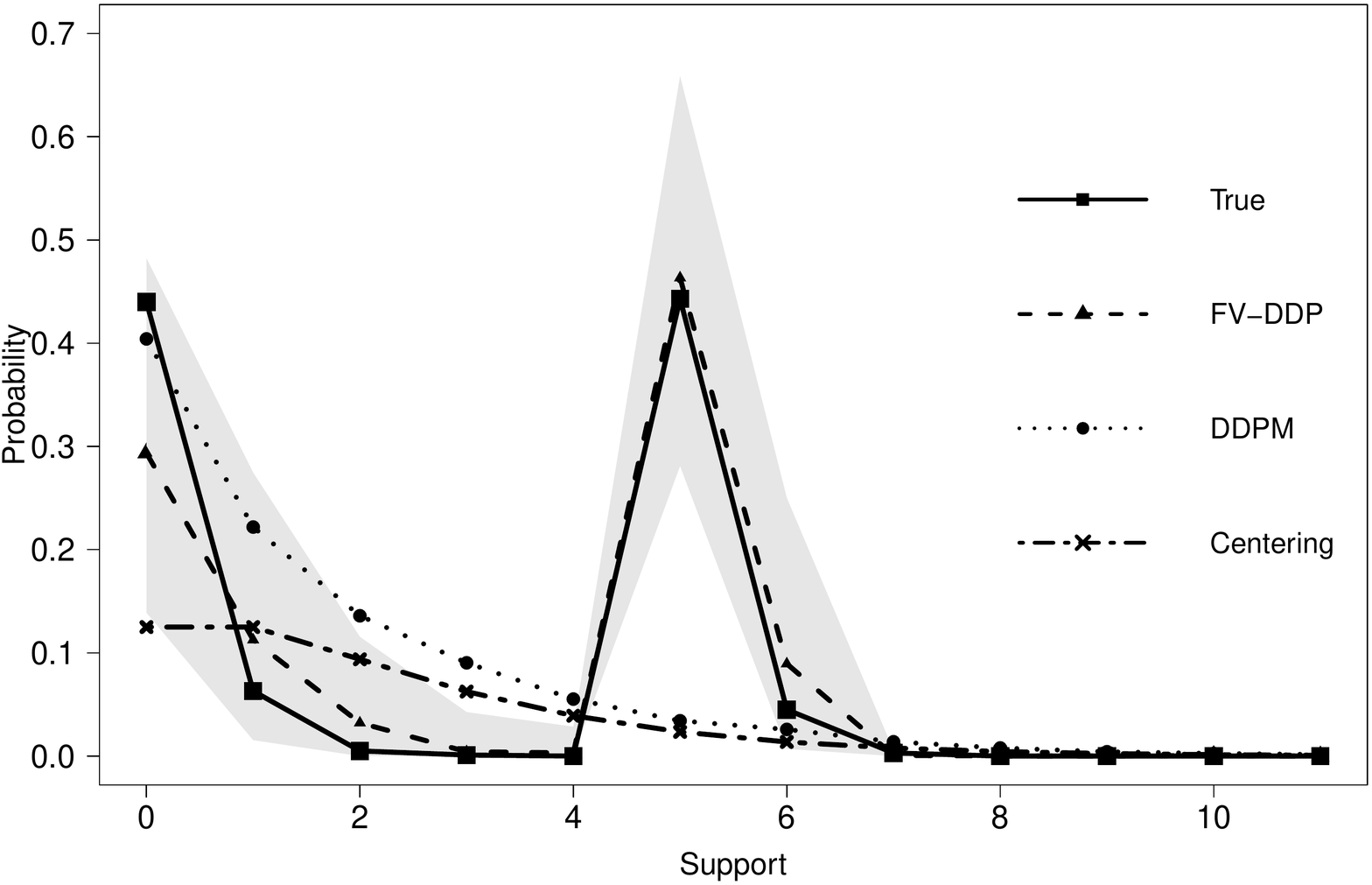}
\begin{quote}
\caption{\scriptsize One- (left) and two-step-ahead prediction (right) based on 5 data collection times, with $95 \%$ pointwise credible intervals.}
   \label{fig:subfig_1}
\end{quote}
\end{figure}


\subsection{Karnofsky score data}

We consider the dataset \emph{hodg} used in \cite{Klein}, which contains records on the time to death or relapse and the Karnofsky score for 43 patients with a lymphoma disease. The Karnofsky score (KS) is an index attributed to individual patients, with higher values indicating a better prognosis. 

In the framework of model \eqref{model}, we take the times of death or relapse as collection times and let the KS of the survivors at each time be the data. We aim at predicting the future distribution of the KS among the patients who are still in the experiment at that time, which would be an indirect assessment of the effectiveness of the score in describing the patients' prognosis. We also include censored observations (patients leaving the experiment for reasons different from death or relapse), without having them trigger a collection time. {The FV-DDP appears as the ideal modeling tool in this framework since it includes a probabilistic mechanism that accounts for the reduced number of observations through different time points. }

We train the model up to $42$, $108$ and $406$ days after the start of the experiment, and we make predictions $28$, $112$ and $144$ days ahead, respectively.
As regards the prior, we put a uniform distribution on the observed scores (note that new score values cannot appear along the experiment) and we uniformly randomize $\theta$ over $\{.5,1,1.5, \dots, 15\}$, analogously to Section \ref{sec: sim data}. Given the results of the previous subsection for different approaches to selecting $\sigma$, here, after transforming the lags in annual, we proceed by selecting $\sigma$ for each value of $\theta$ by maximizing the probability that the death process makes the right number of transitions in the desired laps of time.
Some of the selected values for $\sigma_1,\sigma_{2}, \sigma_3$ for the three different trainings, depending on $\theta$, are shown in Table \ref{table 3}.

\begin{table}[h!]
\begin{center}
\begin{tabular}{cccccccc}\hline
$\theta$    & .5 & 1 & 1.5 & $\cdots$& 29 & 29.5 & 30  \\\hline
$\sigma_1$   & 0.4947 & 0.4913  & 0.4885 & $\cdots$& 0.3235 & 0.3266 & 0.3228 \\
$\sigma_2$   & 0.6059 & 0.6014 & 0.5696 &$\cdots$& 0.3684 & 0.3130 & 0.3361 \\
$\sigma_3$   & 0.6149 & 0.6150  &0.5789 &$\cdots$& 0.3063 & 0.3018 & 0.2901 \\\hline
\end{tabular}
\end{center}
\begin{quote}
\caption{\scriptsize Choice of $\sigma$ for some values of $\theta$ for the three trainings.}\label{table 3}
\end{quote}
\vspace{-5mm}
\end{table}

Figure \ref{fig: score} shows the three predictions of the scores distribution.
Coherently with the intuition, as the experiment goes by, individuals with higher KS become predominant: from $70$ to $230$ days the predicted weight associated to a score of $90$ increases of more than $10 \%$, and similarly for $100$. However the distribution of the scores remains pretty stable, apart from the lowest values, meaning that the highest scoring patients actually had much better prognoses, as showed by the third prediction. 

\begin{figure}[t!]
\centering
\includegraphics[width=.49\textwidth]{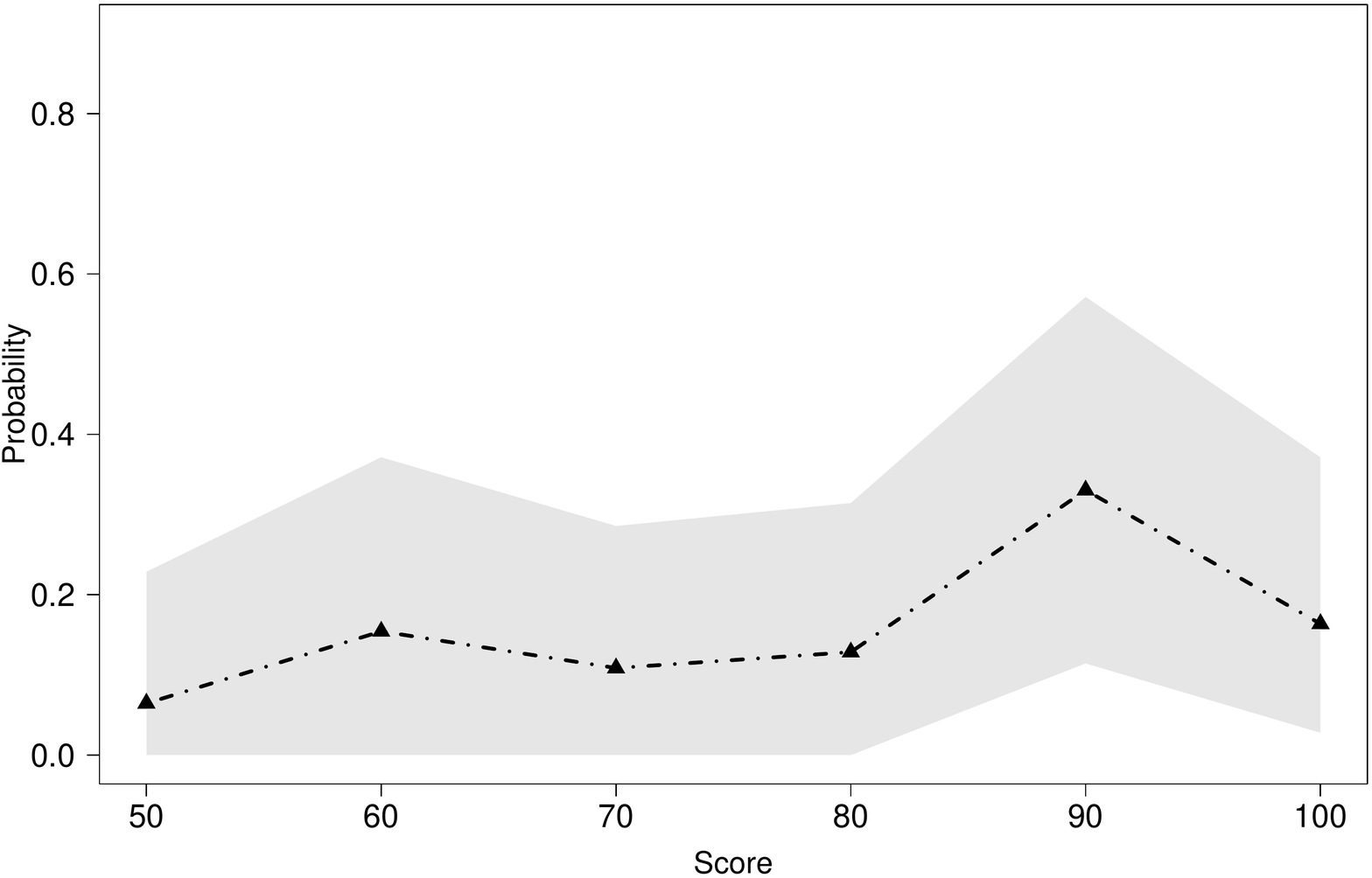} 
\includegraphics[width=.49\textwidth]{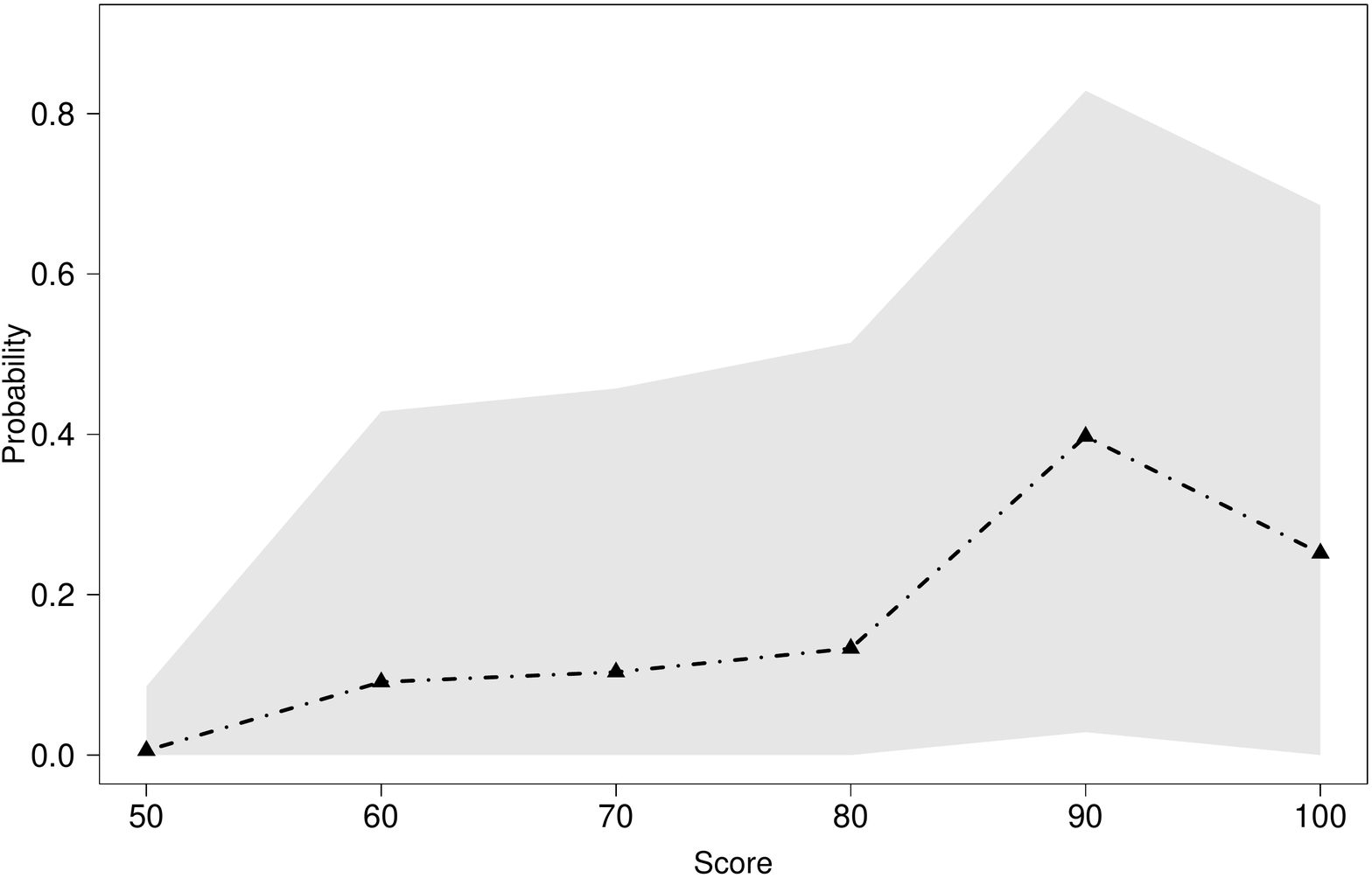} \\
\includegraphics[width=.49\textwidth]{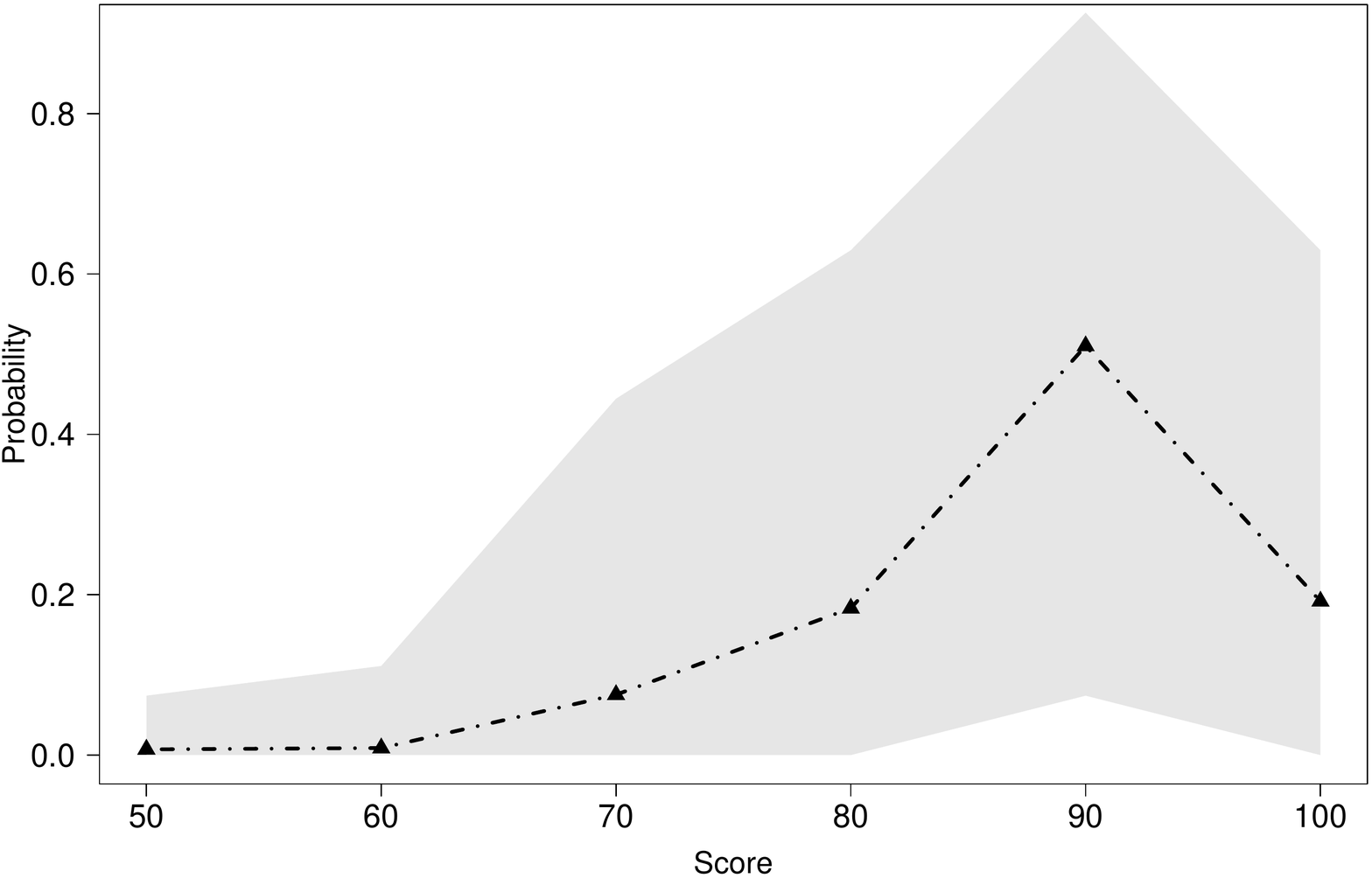}
\includegraphics[width=.49\textwidth]{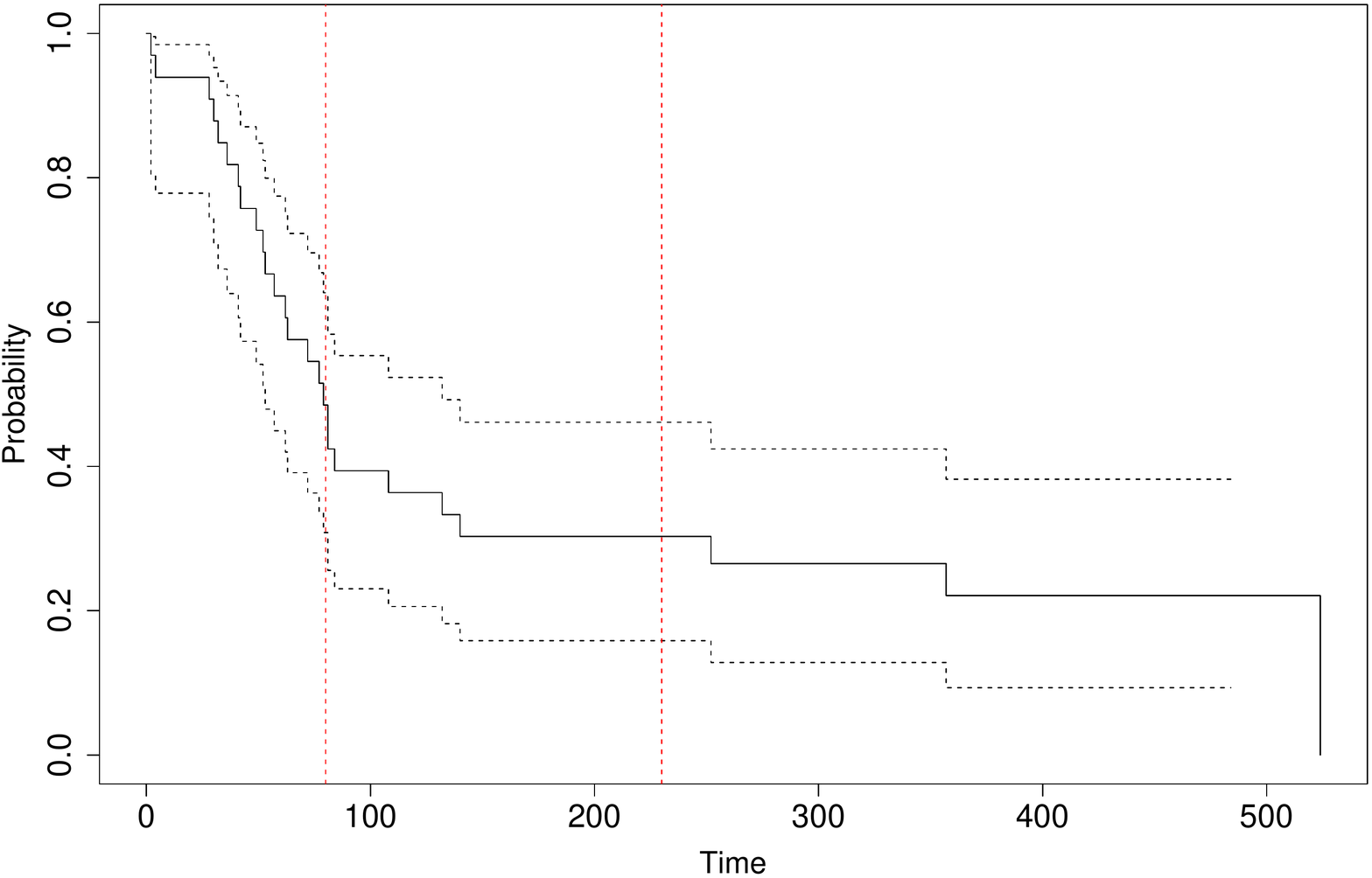}
\begin{quote}
\caption{\scriptsize From top left: pmf prediction at $70$, $230$ and $550$ days after the experiment. Bottom right: Kaplan-Meyer estimate of the survival times up to time $550$.}\label{fig: score}
\end{quote}
\end{figure}

These findings are consistent with the Kaplan-Meyer estimate of the survival function, shown in the bottom right panel, which decreases rapidly between $70$ and $230$ and flattens after that point, implying that the FV-DDP prediction adapted to the periods of quick change in the underlying distribution and periods of relative steady behaviour.


\section{Discussion}

We have derived the predictive distribution for the observations generated by a class of dependent Dirichlet processes driven by a Fleming--Viot diffusion model, which can be characterised as a time-dependent mixture of P\'olya urns, and described the induced partition structure together with practical algorithms for exact and approximate sampling of these quantities. 
An upside of inducing the dynamics through a FV process is that one can implicitly exploit the rich and well understood underlying probabilistic structure in order to obtain manageable closed-form formulae for the quantities of interest. This ultimately relies on the duality with respect to Kigman's coalescent, which was first used for inferential purposes in \cite{PR14}. 

{The approach we have described yields dependent RPMs with almost surely discrete realisations. While such a feature perfectly fits the specific illustrations we have discussed, it is not suited to draw inferences with continuous data.} An immediate and natural extension of the proposed model, which accommodates continuous outcomes would be to consider dependent mixtures of continuous kernels, whereby the observation $y$ from the RPM at time $t$ becomes a latent variable acting as parameter in a parametric kernel $f(z|y)$. This approach would be in line with the extensive Bayesian literature on semi-parametric mixture models, which has largely used the DP or its various extensions as mixing measure. It remains however a non trivial exercise to derive in this framework the corresponding formulae for prediction, which we will leave for future investigation.


\section*{Acknowledgements}

The second and third authors are partially supported by the Italian Ministry of Education, University and Research (MIUR) through PRIN 2015SNS29B. The second author is also supported by MIUR through ``Dipartimenti di Eccellenza'' grant 2018-2022.
Helpful discussions with Amil Ayoub are gratefully acknowledged by the first author.


\section*{Appendix}

\section*{Additional results}

\begin{lemma}\label{lemma death process transition}
The transition probabilities $p_{|\mm|, |\nn|} (t)$ in \eqref{death process transitions} equal $e^{- \lambda_{|\mm|}t}$ when $\nn=\mm$ and
\begin{equation}\nonumber
\left( \prod_{h = 0}^{|\mm-\nn|-1}\lambda_{|\mm|-h} \right)(-1)^{|\mm-\nn|} \sum_{k = 0}^{|\mm-\nn|}\frac{e^{-\lambda_{|\mm|-k}t}}{\prod_{0 \leq h \leq |\mm-\nn|, h \neq k}(\lambda_{|\mm|-k}- \lambda_{|\mm|-h})},
\end{equation}
when $\oo < \nn \leq \mm$, where $\lambda_n = n(\theta+n-1)/2$.
\end{lemma}
\begin{proof}
See \cite{PRS16}, Lemma 4.1.
\end{proof}

\begin{lemma}\label{eppf_lemma}
Assume \eqref{model} and \eqref{oldresult}. Let $I_1, \dots, I_q$, with $q \in \mathbb{N}$ be a partition of $\{1, \dots, q\}$ and let $m_j = |I_j|$. Then the distribution of the partition of $\{1, \dots, q\}$ induced at time $T+t$ is
\begin{equation}\nonumber
\begin{aligned}
p(m_1, \dots , &m_q) =\, \sum_{j = 0}^{\hat{q} } \binom{\hat{q} }{j}\int \dots \int \, \Bigg[ \prod_{s = 0}^{j-1} \sum_{i_j \in \mathcal{K} \backslash \{i_0, \dots , i_{j-1} \} }C_{i_j,k_s} \\
& \prod_{h = 1}^{m_{s+1}-1}\left[ C_{i_j,k_s+h} + B_{k_s+h}\frac{h}{k_s+h} \right] \\
& \quad \prod_{s = j}^{q-1} A_{k_s}\prod_{h=1}^{m_{s+1}-1}B_{k_s+h}\frac{h}{k_s+h} \Bigg] \d P_{i_0}\dots \d P_{i_{j-1}}\d P_0 \dots \d P_0
\end{aligned}
\end{equation} 
with $A_k$, $B_k$ and $C_{i,k}$ as in \eqref{pred_notation}, $\hat{q} =\min \{q, \dist \}$ and $k_i = \sum_{j = 1}^im_j$, while $\mathcal{K} = \{ 1, \dots, \dist \}$ and $P_{i_l}$ is the empirical of past values excluding the already {observed ones}, i.e it is the empirical of the past data points denoted by the set $\mathcal{K} \backslash \{i_0, \dots, i_{l-1} \}$.
\end{lemma}

\begin{proof}
We want to compute $\P \left( \mathcal{P}_m = \left \{ \mathcal{I}_1, \dots , \mathcal{I}_q \right\} \right)$ where $\sum_{j = 1}^qm_j = m$. For simplicity we consider the vector $(k_0, k_1,k_2, \dots, k_q) = (0, m_1, m_1+ m_2, \dots , \sum_{j = 1}^qm_j )$.
Referring to \eqref{pred_notation}, conditioning on the new observations, for any group $m_{s+1}$ we have two possibilities:
\begin{itemize}
\item similarly to the DP case, with probability $A_{k_s}$ the new value $z^*$ comes from $P_0$ and the next $m_{s+1}-1 = k_{s+1}-k_s - 1$ observations will be equal to $z^*$ with probability $\prod_{h=1}^{m_{s+1}-1}B_{k_s+h}\frac{h}{k_s+h}$;

\item with probability $C_{i,k_s}$ the new value $z^*$ is $y^*_i$ and the next $m_{s+1}-1 = k_{s+1}-k_s - 1$ observations will be equal to $z^*$ with probability $\prod_{h=1}^{m_{s+1}-1} \left[ C_{i,k_s+h} + B_{k_s+h}\frac{h}{k_s+h} \right]$.
\end{itemize}
In the second case we have to sum over all different past observations $i$ that were not taken into consideration in the groups before. 
Then, by exchangeability, we can assume that the first $j$ groups sample a new observation from $P_{\nn}$, i.e.~the past observations, while the others from $P_0$. Defining $\hat{q} =\min \{q, \dist \}$, we have that the probabilities of multiplicities $m_1, \dots , m_q$ with $j$  groups from $P_{\nn}$ and $\hat{q}  - j $ from $P_0$, still conditioned on the new observations, is
\begin{equation}\nonumber
\begin{aligned}
\prod_{s = 0}^{j-1} \sum_{i_j \in \mathcal{K} \backslash \{i_0, \dots , i_{j-1} \} }&C_{i_j,k_s} \prod_{h = 1}^{m_{s+1}-1}\left[ C_{i_j,k_s+h} + B_{k_s+h}\frac{h}{k_s+h} \right]\\
& \prod_{s = j}^{q-1} A_{k_s}\prod_{h=1}^{m_{s+1}-1}B_{k_s+h}\frac{h}{k_s+h}
\end{aligned}
\end{equation}  
where we have used $i_0, i_1, \dots, i_j$ to highlight that we cannot consider twice the same past observation. The latter in turn is the same as considering the first $j$  groups from $P_n$ and the last $\hat{q}  - j $ from $P_0$ times the binomial coefficient $\binom{\hat q}{j}$. Then the thesis is obtained by integrating out $j$ and the new observations.

\end{proof}

\section*{Proof of Propositions}

\subsection*{Proof of Proposition 1}

In this proof we use the same notation of \cite{BJQ12} and denote by $G(t)$ the FV-DDP, i.e. $G(t) = X_t$. We also emphasise the elementary event $\omega\in \Omega$ by writing $G(t,\omega)$. By Eq.~3 in \cite{BJQ12}, it suffices to show that for $\epsilon > 0$, $N \in \mathbb{N}$ and $(t_1, \dots , t_N) \in \mathbb{R}_+^N$ we have
\begin{equation}\label{supp}
\P \left\{ \omega \in \Omega \, : \,  \left[ G(t_i,w)(A_0), \dots , G(t_i,w)(A_k) \right] \in B(\textbf{s}_{t_i},\epsilon), i = 1,\dots, N \right\} > 0. 
\end{equation}
Here:
\begin{itemize}
\item $A_0, \dots, A_k$ is a partition of $\Y$, with $A_i$ a measurable set with $P_0$-null boundary;
\item $B(\textbf{s}_{t_i},\epsilon) = \{(w_0, \dots , w_k) \in \Delta_k : \, w_{(t_i,j)} - \epsilon < w_j < w_{(t_i,j)} + \epsilon, j = 0, \dots , k \}$, with $\Delta_k = \{ (w_0, \dots , w_k): \, w_i \geq 0, i = 0, \dots , k, \sum_{i = 0}^kw_i = 1 \}$ the $k$-simplex.
\item $\textbf{s}_{t_i} = (w_{(t_i,0)}, \dots , w_{(t_i,k)}) = \left( Q_{t_i}(A_0), \dots, Q_{t_i}(A_k) \right) \in \Delta_k$.
\item $Q_{t_i}$, $, i = 1, \dots, N$ is a probability measure absolutely continuous with respect to $P_0$.
\end{itemize}
As is well known, projecting a Dirichlet process $\Pi_\alpha$ on a partition $A_0, \dots, A_k$ yields a $k$-dimensional Dirichlet density $\pi_\alpha$ with parameters $(\alpha(A_0), \dots , \linebreak\alpha(A_k))$. Similarly, projecting a FV process yields a a $k$-dimensional Wright-Fisher (WF) diffusion, which is reversible and stationary with respect to $\pi_\alpha$ (cf.~\cite{Dawson}). Consistently with \eqref{transition}, the transition density of the WF is given by:
\begin{equation}\nonumber
P_t(\textbf{x}, \d\textbf{x}') = \sum_{m = 0}^\infty d_m(t) \sum_{\mm \in \mathbb{Z}^{k+1}_+ : |\mm| = m} \binom{m}{\mm} \textbf{x}^{\mm}\pi_{\alpha+\mm}(\textbf{x}')\d \textbf{x}'.
\end{equation}
Then we can rewrite \eqref{supp} as:
\[
\int_{B(\textbf{s}_{t_1},\epsilon)} \dots \int_{B(\textbf{s}_{t_N},\epsilon)} \pi_\alpha(\textbf{x}_1)P_{t_2-t_1}(\textbf{x}_1, \textbf{x}_2) \dots P_{t_N-t_{N-1}}(\textbf{x}_{N-1}, \textbf{x}_N) \, \d \textbf{x}_1 \dots \d \textbf{x}_N
\]
Since $B(\textbf{s}_{t_1},\epsilon)$ has strictly positive Lebsegue measure, we just need to show that the integrand is strictly bigger than $0$ for any $(\textbf{x}_1, \dots ,\textbf{x}_N) \in B(\textbf{s}_{t_1},\epsilon) \times \dots \times B(\textbf{s}_{t_N},\epsilon)$. Clearly $\pi_\alpha(\textbf{x}_1) > 0$ for any $\textbf{x}_1 \in B(\textbf{s}_{t_1},\epsilon)$. For what concerns $1 < j \leq N$, we have:
\[
P_{t_j-t_{j-1}}(\textbf{x}_{j-1}, \textbf{x}_j) \geq d_0(t_j-t_{j-1}) \pi_\alpha(\textbf{x}_j) > 0, \quad \forall \textbf{x}_j \in B(\textbf{s}_{t_j},\epsilon),
\]
which completes the proof.

\subsection*{Proof of Proposition 2}

Conditioning on the random measure $X_{T+t}$ at time $T+t$ yields
\begin{equation}\label{start}
\begin{aligned}
\P \big(&\,Y_{T+t}^{k+1} \in A \mid \YY_{0:T}, Y_{T+t}^{1:k} \big)\\
=&\, \E \Big[ \P \big(Y_{T+t}^{k+1} \in A 
\mid X_{T+t}, \YY_{0:T}, Y_{T+t}^{1:k} \big) \mid \YY_{0:T}, Y_{T+t}^{1:k} \Big] \\
=&\, \E \Big[ \P \big(Y_{T+t}^{k+1} \in A \mid X_{T+t} \big) \mid \YY_{0:T}, Y_{T+t}^{1:k} \Big] \\
=&\,  \E \Big[ X_{T+t}(A)  \mid \YY_{0:T}, Y_{T+t}^{1:k} \Big] 
\end{aligned}
\end{equation} 
where the second equality follows from the conditional independence of the observations given the signal; cf.~\eqref{model}. From \eqref{oldresult}, eq.~(3.7) in \cite{PRS16} implies that $X_{T+t} \mid \YY_{0}, \dots , \YY_T$ is the mixture of Dirichlet processes
\begin{equation}\nonumber
\sum_{\nn \in L(\MM)}p_t\left(\MM, \nn \right)  \Pi _{\alpha + \sum_{i=1}^{\dist}n_i \delta_{y_i^*}}.
\end{equation} 
 By linearity of the expectation and using \eqref{DPpredictive}, when $k=0$ the RHS of \eqref{start} reads
\[
\begin{aligned}
\sum_{\nn \in L(\MM)}&\,p_t\left(\MM, \nn \right) \E \left[ \Pi _{\alpha + \sum_{i=1}^{\dist}n_i \delta_{y_i^*}} \left( A \right) \right] = \\
=&\, \sum_{\nn \in L(\MM)}p_t\left(\MM, \nn \right) \biggl[ \frac{\theta}{\theta + |\nn|} P_0 \left(A \right) + \frac{|\nn|}{\theta + |\nn|} \sum_{i=1}^{\dist}n_i \delta_{y_i^*} \left(A \right) \biggr] \\
=&\, \sum_{\nn \in L(\MM)}p_t(\MM, \nn )  
\frac{\theta}{\theta + |\nn|}  P_0 \left( A \right) + \sum_{\nn \in L(\MM)}p_t\left(\MM, \nn \right) \frac{|\nn|}{\theta + |\nn|} P_\nn,
\end{aligned}
\]
which is \eqref{predictive 1} with $k = 0$. When $k > 0$, using again \eqref{DPpredictive} and the conjugacy property of mixture of Dirichlet processes, the RHS of \eqref{start} reads
\begin{equation}\nonumber
\begin{aligned}
\E \Bigg[&\, \sum_{\nn \in L(m)}p_{t}(\MM,\nn)\Pi_{\alpha + \sum_{i = 1}^{\dist}n_i\delta_{y_i^*}} \bigg\vert Y_{T+t}^{1:k} \Bigg] \\ 
=&\, \sum_{\nn \in L(\MM)}p^{(k)}_t(\MM,\nn) \E \bigg[ \Pi _{\alpha + \sum_{i=1}^{\dist}n_i \delta_{y_i^*} + \sum_{j=1}^k \delta_{y_j}}\bigg] \\
\end{aligned}
\end{equation} 
yielding \eqref{predictive k+1}.

\subsection*{Proof of Proposition 3}
Denote
\begin{equation}\nonumber
\E_0[Y] = \int y \, P_0(\d x), \quad 
\E_0[Y^2] = \int y^2 \, P_0(\d x).
\end{equation} 
We need to compute
\[
\text{Corr}(Y_{t},Y_{t+s}) = \frac{\text{Cov}(Y_{t},Y_{t+s})}{\E_0[Y^2]-\E_0^2[Y] } = \frac{\E[Y_tY_{t+s}]-\E_0^2[Y]}{\E_0[Y^2]-\E_0^2[Y] },
\]
where
\[
\E[Y_tY_{t+s}] = \int y_ty_{t+s} \, P(\d y_t, \d y_{t+s}).
\]
From Proposition \ref{prop: main}, we can write the joint distribution using the chain rule, yielding
\[
P(\d y_t, \d y_{t+s}) = P_0(\d y_t)\left[\left(1-e^{-\frac{\theta}{2}s}\right)P_0(\d y_{t+s}) + \frac{\theta e^{-\frac{\theta}{2}s}}{\theta+1}P_0(\d y_{t+s}) + \frac{e^{-\frac{\theta}{2}s}}{\theta+1}\delta_{y_t}(\d y_{t+s})  \right]
\]
which in turn gives
\[
\E[Y_tY_{t+s}] = \left(1-e^{-\frac{\theta}{2}s}+\frac{\theta e^{-\frac{\theta}{2}s}}{\theta+1}  \right)\E_0^2[Y] + \frac{e^{-\frac{\theta}{2}s}}{\theta+1}\E_0[Y^2].
\]
Consequently
\[
\text{Cov}(Y_{t},Y_{t+s}) = \frac{e^{-\frac{\theta}{2}s}}{\theta+1}\left(\E_0[Y^2]-\E_0^2[Y] \right)
\]
from which the result follows.

\subsection*{Proof of Proposition 4}

Denote by $P_{0,k}$ the RHS of \eqref{DPpredictive}. We have to prove that
\begin{equation}\label{predictive convergence}
\left \lvert \P \big( Y_{T+t}^{k+1} \in A \mid \YY_{0:T}, Y_{T+t}^{1:k} \big) - P_{0,k}(A) \right \rvert \to 0, \quad \forall A \in \mathcal{B}\left(\mathcal{Y}\right)
\end{equation} 
as $t \to \infty$. Using the triangle inequality the LHS of \eqref{predictive convergence} is smaller than:
\[
\left \lvert \sum_{\textbf{n} \in L(\MM)} p_t^{(k)}(\MM, \nn)\frac{-|\textbf{n}|}{\theta+ |\textbf{n}|+k}P_{0,k}(A) \right\rvert+
\left \lvert \sum_{\textbf{n} \in L(\MM)} p_t^{(k)}(\MM, \nn)\frac{|\textbf{n}|}{\theta+ |\textbf{n}|+k}P_\textbf{n}(A) \right \rvert
\]
Note now that the time-dependence of \eqref{predictive k+1} is ultimately due to $p_{|\mm|, |\nn|}(t)$ in \eqref{death process transitions}. These are the transition probabilities of a one-dimensional death process on $\mathbb{Z}_{+}$ which jumps from $m$ to $m-1$ at infinitesimal rate $\lambda_m = m(\theta+m-1)/2$. It can be easily verified that, as $t \to \infty$, we have $p_{|\mm|, 0}(t)\rightarrow 1$ for any $\mm$ and $p_{|\mm|, |\nn|}(t)\rightarrow 0$ for any $\oo<\nn\le \mm$, and similar statement holds for \eqref{death process transitions from M}. Then, denoting $B_{1},B_{2}$ the two sums in the previous display respectively, we have
\[
0 \leq \max\{B_{1},B_2\} \leq \sum_{\textbf{n} \in L(\MM)}p_t^{(k)}(\MM, \nn)\frac{|\textbf{n}|}{\theta+ |\textbf{n}|+k}  \to 0
\]
which implies \eqref{predictive convergence}, as desired.
%

\subsection*{Proof of Proposition 5}

By de Finetti's Representation Theorem $P_k \to P^*$ as $k \to \infty$, with $P^{*}$ being the De Finetti measure of the sequence $\left(Y_{T+t}^k \right)_{k \geq 1}$. Moreover, recalling that $L(\MM)$ is a finite set, we have:
\[
\lim _{k \to \infty}\sum_{\nn \in L(\MM)}p^{(k)}_t(\MM,\nn)  \frac{k}{\theta + |\nn| + k } = \lim _{k \to \infty}\sum_{\nn \in L(\MM)}p^{(k)}_t(\MM,\nn) = 1
\]
 As regards the other two components of \eqref{predictive k+1} we have
\begin{equation}\nonumber
0 \leq \lim _{k \to \infty}\sum_{\nn \in L(\MM)}p^{(k)}_t(\MM,\nn)  \frac{1}{\theta + |\nn| + k } \leq \lim _{k \to \infty} \sum_{\nn \in L(\MM)} \frac{1}{\theta + |\nn| + k } = 0 
\end{equation}
and we have the result.

\end{document}